%% file: main.tex
\documentclass[12pt]{article}
\usepackage{graphicx} 
\usepackage{amsmath,amssymb,amsthm}
\usepackage{array}
\usepackage{lmodern}
\usepackage{iftex}
\usepackage[round,authoryear]{natbib}
\usepackage[hidelinks]{hyperref}

\newtheorem{theorem}{Theorem}[section]
\newtheorem{proposition}[theorem]{Proposition}
\newtheorem{lemma}[theorem]{Lemma}

\usepackage{authblk}

\begin{document}

\def\spacingset#1{\renewcommand{\baselinestretch}%
{#1}\small\normalsize} \spacingset{1}


\title{From Cumulative Weights to Marginal Density Ratios: Per-Protocol Estimation in Sequential Target Trial Emulation}
\renewcommand\Authands{, }
  \author[1]{Zern Ke}
  \author[1]{Mingshi Cui}
  \author[2]{Feng Dai}
  \author[2]{Birol Emir}
  \author[1]{Javier Cabrera}
  \author[2]{Demissie Alemayehu}
  \affil[1]{Department of Statistics, Rutgers University, New Brunswick}
  \affil[2]{Pfizer Inc, New York, NY, USA}

\date{ }
\maketitle

\begin{abstract}
Sequential target trial emulation evaluates eligibility at multiple
baseline times to emulate a sequence of randomized trials using
observational data. Estimating per-protocol effects in this setting is
challenging because treatment deviations and loss to follow-up induce
selection among individuals who remain observed and adherent over time.
Conventional inverse-probability methods address this selection using
cumulative weights constructed from estimated adherence and censoring
probabilities, but these weights can be highly variable, leading to
unstable and imprecise effect estimates. We propose a
different approach based on marginal density ratios (MDRs). 
The MDR directly compares the state distribution
among individuals who would remain event-free under a target treatment
strategy with the corresponding distribution among observed-adherent
individuals. We use longitudinal g-computation to generate the target
risk sets and a probabilistic classifier to estimate density ratios for
reweighting the observed outcomes. Building on this approach, we also
develop a doubly robust extension. Favorable performance across the
simulation study suggests that MDR weighting is a promising alternative
to cumulative longitudinal weights when its identification assumptions
are plausible.
\end{abstract}

\newpage
\spacingset{1.5} 

\input{introduction.tex}

\input{related_work}

\input{methodology}

\input{simulation_study}

\input{discussion}

\clearpage
\spacingset{1.1} 
\bibliographystyle{custom-apalike}
\bibliography{reference}

\clearpage
\appendix
\input{appendix}

\end{document}

%% file: introduction.tex
\section{Introduction}

\subsection{Background and Motivation}

Randomized controlled trials provide the most direct way to compare
treatment strategies, but a suitable trial is not always available.
Trials can be expensive, require years of follow-up, raise ethical
concerns, or enroll patients who differ from those seen in routine
practice. Longitudinal observational data, including electronic health
records, insurance claims, and clinical registries, can provide
important additional evidence. These sources often contain repeated
measurements of treatment, patient characteristics, and clinical
outcomes. However, drawing causal conclusions from these data is
difficult because treatment is not randomly assigned and may change as
a patient's health changes
\citep{hernan2016targettrial,su2024trialemulation}.

The target trial framework helps researchers design an observational
study around the randomized trial they would ideally conduct. They
first specify a protocol with the eligibility criteria, treatment
strategies, assignment procedures, start and end of follow-up, outcome,
causal contrast, and analysis plan. They then use the available
observational data to emulate this protocol as closely as possible.
Aligning eligibility, treatment assignment, and the start of follow-up
at a common time zero can reduce design errors such as immortal-time
bias. It also makes the causal question easier to understand and
evaluate \citep{hernan2016targettrial}.

A longitudinal database may contain more information than can be used
in a trial emulated from only one baseline. Treatment decisions are
often made repeatedly, and a patient may be eligible to start treatment
at several times. An analysis based on a single baseline may therefore
miss later treatment opportunities, include few treatment initiators,
and depend heavily on the chosen cohort-entry time
\citep{keogh2023sequentialtrials}.

Sequential target trial emulation addresses this limitation by
emulating a series of trials with different calendar starting times. The
trials follow the same general protocol, but eligibility is reassessed
at the start of each trial. An eligible individual enters a new
person-trial and is followed from that treatment decision. The same
individual may contribute to several trials while remaining eligible.
This design reflects clinical practice, in which patients may have
repeated opportunities to initiate treatment or remain untreated
\citep{keogh2023sequentialtrials,su2024trialemulation}.

This process creates aligned person-trial records from one longitudinal
cohort. Each record has its own baseline covariates, treatment
assignment, and follow-up time. The records can then be pooled across
sequential trials. In this paper, each eligible person-trial receives equal
baseline weight. At a prespecified horizon \(\tau\), the primary
estimand is the pooled per-protocol cumulative-risk difference: the risk
of the outcome by \(\tau\) under sustained treatment minus the
corresponding risk under sustained no treatment.
The target population therefore consists of eligible person-trial
entries, not unique individuals. This design uses repeated eligibility
and treatment opportunities while maintaining a clear time zero for
each entry. It preserves the correct order of eligibility, covariate
measurement, treatment, and outcome follow-up. This pooled person-trial
structure is the basis of our analysis.

Sequential target trial emulation builds on methods for longitudinal
causal inference. The longitudinal g-formula can be used to estimate the
effects of sustained treatment strategies when treatment and
confounders change over time
\citep{robins1986gformula}. Marginal structural models and
inverse-probability weighting address treatment-confounder feedback.
This feedback occurs when a time-varying covariate affects future
treatment and the outcome but is itself affected by earlier treatment
\citep{robins2000msm}. Other methods include parametric g-computation
and structural nested models
\citep{daniel2013timedependent,young2011parametricgformula}, as well as
targeted and doubly robust estimators
\citep{lendle2017ltmle,bang2005doublyrobust}. The target trial framework
provides a study-design structure for these methods, and sequential
emulation extends that structure to repeated treatment opportunities.
Section~\ref{sec:related-work} provides a fuller review.

Good study design does not by itself solve the estimation problem,
especially for per-protocol effects. During follow-up, individuals may
stop or switch treatment, otherwise deviate from the target strategy,
or be lost to follow-up. Those who remain event-free, observed, and
adherent may differ from those who would remain at risk if everyone
followed the target strategy. Restricting the analysis to
observed-adherent individuals can therefore produce bias
\citep{keogh2023sequentialtrials,cole2008constructing}.

A common solution is cumulative inverse-probability weighting based on
the estimated probabilities of adherence and continued observation.
Unstabilized weights multiply inverse conditional probabilities across
follow-up visits. Stabilized weights include numerator probabilities
intended to reduce weight variability
\citep{robins2000msm,cole2008constructing}. However, stabilization does
not remove the underlying product of many probabilities. Weights can still
be highly variable when adherence is limited, treatment strongly
depends on patient history, follow-up is long, or treatment overlap is
weak. These methods also reconstruct the target population one
treatment and censoring decision at a time. Our method instead focuses
on the state distribution at each follow-up interval among individuals
who would remain event-free under the target strategy. We ask whether
this target risk-set distribution can be linked directly to the
observed-adherent risk-set distribution, without recovering it through
a long product of sequential probabilities.

Off-policy evaluation has used marginal state-distribution ratios in
place of products of trajectory-level importance ratios
\citep{liu2018breakinghorizon,xie2019mis}. We adapt this idea to the
counterfactual target and observed-adherent risk sets of sequential
target trials. We call the resulting method sequential target trial
emulation using marginal density ratios (STTE--MDR). STTE--MDR still
requires longitudinal models. It uses fitted event and state-transition
models to generate strategy-specific target risk sets, then uses a
classifier to estimate the direct density ratio. Thus, it avoids
modeling and multiplying a sequence of adherence and censoring
probabilities.

This reduction from full histories to current states requires several
conditions. In addition to the standard causal identification
assumptions, the state must retain the information needed to model the
conditional event risk and, when states are propagated over time, their
evolution among survivors.
The conditional event-risk relationship estimated in the
observed-adherent risk set must also apply to the target risk set.
Finally, target states must have support in the observed-adherent state
distribution, and the target-state generator must consistently estimate
the target risk-set distributions.

\subsection{Contributions}

This paper makes four main contributions. First, we frame per-protocol
estimation in sequential target trials as a risk-set transport problem.
The goal is to transport the observed-adherent state distribution to
the strategy-specific target state distribution. Second, we define a
direct marginal density ratio that captures the effects of earlier
adherence and censoring on the current risk-set distribution, without
forming cumulative products of treatment and censoring probabilities.
The population MDR is the conditional average of the full-history ratio
among histories with the same current state. Its population variance is
therefore no greater than that of the full-history ratio. This result
does not guarantee lower finite-sample variance for the estimated
treatment effect. Third, we give an implementable STTE--MDR algorithm that
combines longitudinal target-state simulation, follow-up-balanced
classifier-based density-ratio estimation, cell-normalized weighting,
and standardization of discrete-time event hazards. Fourth, we
establish identification results, relate the direct MDR to cumulative
full-history weights, and develop a doubly robust extension. Conditional
on correct generation of the target risk-set distribution, this
extension remains consistent if either the outcome regression or the
MDR is correctly specified. We evaluate the proposed estimators under
poor adherence, limited overlap, strong confounding, small samples,
rare outcomes, and nonlinear model structures scenarios.

\subsection{Organization of the Paper}

The remainder of the paper is organized as follows.
Section~\ref{sec:related-work} reviews prior work on sequential target
trials and major longitudinal causal
methods, including weighting and balancing, g-computation, density-ratio,
targeted, and doubly robust estimation.
Section~\ref{sec:methodology} defines the sequential target trial setup,
causal estimands, and identification assumptions. It then develops the
STTE--MDR and doubly robust estimators and presents their theoretical
properties. Section~\ref{sec:simulation-study} describes the simulation
study and compares the proposed methods with existing approaches under
common challenges in longitudinal observational studies.
Section~\ref{sec:discussion} discusses the main findings, practical
implications, limitations, and directions for future research.

%% file: related_work.tex
\section{Related Work}
\label{sec:related-work}

\subsection{Sequential Target Trial Emulation and Per-Protocol Effects}
\label{sec:related-stte}

The methodological foundations of sequential target trial emulation were
established through earlier applications of trial-based analyses to longitudinal
observational data. \citet{hernan2008observational} analyzed observational
follow-up as a sequence of trials, allowing eligible individuals to enter the
analysis at multiple treatment-decision times. \citet{hernan2016targettrial}
subsequently formalized the target trial framework and emphasized the need to
specify the protocol of the hypothetical randomized trial before designing the
observational analysis. Together, these contributions provided the conceptual
basis for modern sequential target trial emulation.

\citet{keogh2023sequentialtrials} provided a detailed methodological comparison of
the sequential trials approach and marginal structural models estimated using
inverse-probability-of-treatment weights. They showed how both approaches can
be used to estimate marginal survival contrasts under sustained treatment
strategies. Their results indicated that sequential trials may produce less
extreme weights and greater efficiency at many follow-up times, although these
advantages are not uniform and depend on the follow-up period and model
specification.

\citet{su2024trialemulation} developed the \texttt{TrialEmulation} package to
support the practical implementation of sequential target trial analyses. The
package constructs person-trial data, estimates weights for treatment switching
and dependent censoring, fits marginal structural outcome models, and estimates
intention-to-treat and per-protocol risk differences.
\citet{limozin2025inference} further studied inference for sequential trial
emulation by comparing confidence intervals based on sandwich variance
estimators, bootstrap procedures, and the jackknife. Despite these developments,
standard weighted implementations of sequential-trial per-protocol analyses,
including those emphasized above, typically rely on cumulative
treatment-adherence and censoring weights. The construction and limitations
of these weights are discussed in the next subsection.

\subsection{Cumulative Weighting and Longitudinal Balancing}
\label{sec:related-weighting}

Marginal structural models were introduced to estimate causal effects when
treatment and confounders vary over time, including settings with
treatment--confounder feedback. In such settings, a time-varying confounder
may affect future treatment and the outcome while also being affected by
earlier treatment. Direct adjustment for this variable in a conventional
outcome model may therefore block part of the treatment effect or introduce
bias. \citet{robins2000msm} introduced inverse-probability weighting for
marginal structural models to address this problem. The unstabilized weights
create a pseudo-population in which treatment no longer depends on the
measured time-varying confounder history. Related censoring weights adjust
for selection when continued observation depends on the measured history.

In longitudinal analyses, treatment and censoring weights are multiplied
across follow-up. A generic unstabilized weight through visit \(k\) is
\begin{equation}
\label{eq:unstabilized_weight}
W_{ik}
=
\prod_{j=0}^{k}
\frac{1}{\Pr\!\left(A_{ij}\mid H_{ij}\right)}
\prod_{j=0}^{k-1}
\frac{1}{
    \Pr\!\left(C_{i,j+1}=0\mid H_{ij},A_{ij}\right)
}.
\end{equation}
Here, \(i\) indexes individuals, \(j\) indexes follow-up visits, and \(k\)
is the current follow-up time. \(A_{ij}\) denotes treatment at visit \(j\),
\(H_{ij}\) is the measured treatment and covariate history available at that
visit, and \(C_{i,j+1}=0\) indicates that the individual remains uncensored
after visit \(j\). The first product adjusts for treatment or adherence, and
the second adjusts for censoring. Conditioning on being uncensored and at risk at the beginning of each visit is
suppressed in Eq.~\ref{eq:unstabilized_weight} for simplicity. In per-protocol
analyses, deviations from the assigned strategy can be handled in different
ways. Some methods represent them through treatment-switching weights
\citep{su2024trialemulation}, whereas others encode them as artificial
censoring at the time of deviation \citep{hernan2008observational}.
Unstabilized weights can become large when any of the estimated probabilities
is small.

\citet{robins2000msm} also presented stabilized weights. These weights add
numerator probabilities based on a smaller set of variables:
\begin{equation}
\label{eq:stabilized_weight}
SW_{ik}
=
\prod_{j=0}^{k}
\frac{
    \Pr\!\left(A_{ij}\mid \bar{A}_{i,j-1},V_i\right)
}{
    \Pr\!\left(A_{ij}\mid H_{ij}\right)
}
\prod_{j=0}^{k-1}
\frac{
    \Pr\!\left(C_{i,j+1}=0\mid \bar{A}_{ij},V_i\right)
}{
    \Pr\!\left(C_{i,j+1}=0\mid H_{ij},A_{ij}\right)
}.
\end{equation}
Here, \(\bar{A}_{i,j-1}\) is the treatment history before visit \(j\), and
\(V_i\) contains the baseline variables used to stabilize the weights and
included in the marginal structural outcome model. Stabilization preserves
treatment and censoring patterns explained by these variables and usually
keeps the weights closer to one. It can therefore reduce variability and
improve precision relative to unstabilized weighting.

Stabilized weights still multiply probabilities across follow-up. They may
therefore remain highly variable when adherence is limited, overlap is weak,
censoring is common, or follow-up is long. Weight truncation is often used to
reduce the influence of extreme values, although it introduces a
bias--variance tradeoff. \citet{cole2008constructing} provided detailed
guidance on weight construction, stabilization, positivity, diagnostics, and
truncation.

Later methods placed greater emphasis on covariate balance.
\citet{imai2015robust} extended the covariate balancing propensity score to
longitudinal treatments. Their method retains sequential treatment models and
cumulative inverse-probability weights, but estimates the model parameters
using balance conditions rather than maximum likelihood alone. At each visit,
the method aims to balance measured covariates across current and future
treatment assignments among individuals with the same past treatment history.
This construction can reduce sensitivity to treatment-model misspecification,
although the weights remain products of sequential treatment probabilities.

\citet{yiu2022joint} jointly calibrated inverse-probability treatment and
censoring weights for marginal structural models. Starting with
cumulative treatment-and-censoring weights, their method applies an
exponential calibration factor chosen to satisfy predefined balance conditions for treatment and censoring simultaneously. After calibration, treatment is less
associated with measured covariate history, and the weighted uncensored sample
better represents the
population that would remain without censoring.

\citet{zhou2020residual} proposed residual balancing as an alternative to
modeling the full treatment-assignment distribution. Their method first
estimates the conditional means of the time-varying confounders given the
earlier history. It then constructs weights that balance the resulting
confounder residuals across later treatment values and relevant treatment
histories. This approach can be easier to apply with continuous treatments
and may be less sensitive to misspecification of the treatment model. Its
validity, however, depends on the conditional-mean models for the confounders
and the selected balance conditions. Residual balancing remains a method for
constructing weights for a marginal structural model.

Kernel optimal weighting provides a more flexible balance-based approach. 
\citet{kallus2021optimal} proposed selecting longitudinal weights by jointly 
minimizing imbalance in time-varying confounders and penalizing weight 
dispersion. The framework can also be extended to informative censoring. 
Its published formulation targets longitudinal marginal structural model
estimation rather than sequential target trial emulation.
Like STTE--MDR, kernel optimal weighting avoids direct estimation of
sequential treatment probabilities and focuses on distributional alignment.
Kernel optimal weighting balances time-varying confounders across observed treatment 
histories for a marginal structural model, with a penalty for unstable weights, whereas STTE--MDR directly aligns 
the observed-adherent state distribution with the strategy-specific target 
risk-set distribution. Thus, the former defines weights through an imbalance
objective, whereas the latter estimates an explicit risk-set change of
measure.

\subsection{Density-Ratio Transport and Marginal State Weighting}
\label{sec:related-density-ratio}

Both covariate-balancing and density-ratio methods seek weights that transport
an observed distribution \(P^{A}\) toward a target distribution \(P^{T}\).
They differ mainly in the procedure used to estimate the weights. Covariate-balancing
methods construct weights by imposing or minimizing imbalance over selected
moments or function classes, whereas density-ratio methods directly estimate
the change of measure \(r(s)=dP^{T}/dP^{A}(s)\). If this ratio is correct, then
\[
E_{P^{A}}\{r(S)f(S)\}=E_{P^{T}}\{f(S)\}
\]
for every integrable \(f\), so the reweighted population reproduces the target
distribution of the chosen state \citep{qin1998densityratio}.

A probabilistic classifier provides a method to estimate the target-to-source density ratio without modeling the two densities separately. Target and source observations are pooled, labeled by their origin, and used to fit the classifier. After adjustment for the sample class proportions, the classifier odds give the target-to-source density ratio \citep{bickel2009covariateshift,menon2016densityratio}. 

Causal transport uses a similar change-of-measure idea to reweight trial participants toward a target population, often using baseline covariates; outcome regression can be added for doubly robust estimation \citep{westreich2017transportability,dahabreh2020extending}. STTE--MDR applies this idea within longitudinal follow-up rather than between external populations. For each strategy and follow-up interval, it transports the observed-adherent risk-set distribution toward the corresponding counterfactual survivor risk-set distribution.

Marginal state weighting also appears in off-policy evaluation. Conventional
importance sampling represents a target policy through products of
action-probability ratios along observed trajectories, which can become
unstable over long horizons. \citet{liu2018breakinghorizon} replaced these
trajectory products with ratios of stationary state-visitation distributions,
and \citet{xie2019mis} developed a finite-horizon method based on
time-specific target-policy state distributions. These studies establish the principle of replacing cumulative trajectory ratios with marginal state-visitation ratios. In off-policy evaluation, marginalization can reduce how quickly variance grows as the time horizon increases, but it is not guaranteed to reduce variance at a fixed finite horizon \citep{liu2020conditional}. 

STTE--MDR applies marginal state weighting to a different pair of
distributions: the counterfactual survivor states under each strategy and the
states observed among individuals who remain adherent and uncensored. Because
these risk sets evolve, it estimates a separate ratio for each strategy and
follow-up interval. The unobserved target states are first generated using
longitudinal transition and event models.

\subsection{Longitudinal g-Computation and Targeted Estimation}
\label{sec:related-gcomp-targeted}

The longitudinal g-formula identifies counterfactual outcomes under sustained
treatment strategies \citep{robins1986gformula,daniel2013timedependent}. Parametric g-computation estimates
it by simulating covariate and outcome trajectories under each strategy
\citep{young2011parametricgformula,mcgrath2020gformula}. Standard applications
average the simulated outcomes to estimate counterfactual risks. STTE--MDR
also retains the survivor states at each interval, using them as the target
risk sets for density-ratio estimation.

Several doubly robust methods combine outcome modeling with a second
adjustment mechanism. \citet{bang2005doublyrobust} developed augmented
estimators that combine sequential outcome regressions with
treatment-and-censoring models. \citet{tran2019longitudinalestimators}
compared implementations based on iterated conditional expectations,
inverse-probability weighting, augmented weighting, and targeted estimation.
For population transport, \citet{dahabreh2020extending} averaged trial-based outcome predictions over the target population and added a weighted correction using trial outcomes. For continuous-time survival outcomes, \citet{chatton2022gcomputation} combined a propensity-weighted outcome model with g-computation. In off-policy
evaluation, \citet{kallus2020drl} combined marginalized density ratios with
value-function models. We propose a doubly robust extension of STTE--MDR,
called DR-STTE--MDR. It uses forward g-computation to generate the target
survivor states, averages predicted interval hazards over those states, and
adds an MDR-weighted residual correction based on the observed-adherent risk
set. 

Longitudinal targeted maximum likelihood estimation updates sequential
outcome regressions using the cumulative treatment-and-censoring mechanism to
target the intervention-specific parameter
\citep{vanderlaan2012longitudinaltmle,petersen2014longitudinaltmle}. It
supports static and dynamic strategies and right censoring and is implemented
in the \texttt{ltmle} package \citep{lendle2017ltmle}.
\citet{limozin2026calibratedltmle} combined this approach with the joint
calibration method of \citet{yiu2022joint} for per-protocol target trial
emulation. 

STTE--MDR brings together sequential target trial emulation, longitudinal
g-computation, and marginal density-ratio estimation. It frames per-protocol
estimation in pooled sequential trials as transporting observed-adherent risk
sets to strategy-specific target survivor risk sets. Rather than constructing
weights from cumulative treatment, adherence, and censoring probabilities, it
generates the target states and directly estimates a density ratio within each
strategy--follow-up cell. Its doubly robust extension,
DR-STTE--MDR, adds an MDR-weighted residual correction to outcome predictions
averaged over the target states. Conditional on consistent generation of the
target risk sets, this estimator remains consistent if either the outcome
regression or the MDR is correctly specified.

%% file: methodology.tex

\section{Methodology}
\label{sec:methodology}

\subsection{Sequential Target Trial Setup and Notation}
\label{subsec:setup-notation}

\subsubsection{Target trial specification}
\label{subsubsec:target-trial-specification}

We emulate the same target-trial protocol at several calendar times.
Let \(m\in\mathcal{M}\) index the resulting emulated trials. In every
trial, the protocol specifies the same eligibility criteria, treatment
strategies, follow-up period, outcome, causal contrast, and analysis
\citep{hernan2016targettrial,keogh2023sequentialtrials}.
Let \(E_{im}=1\) indicate that individual \(i\) satisfies the eligibility
criteria at the baseline of trial \(m\). Because eligibility is
reassessed at the start of each trial, the same individual may contribute
to more than one emulated trial. A candidate trial is included only if at
least one individual is eligible at its starting time.

Within trial \(m\), scheduled visits occur at
\[
  t_{m0}<t_{m1}<\cdots<t_{m\tau},
\]
where \(t_{m0}\) denotes trial baseline and \(t_{m\tau}\) denotes the
end of follow-up. Follow-up intervals are indexed by
\(k=0,\ldots,\tau-1\), with interval \(k\) defined as
\([t_{mk},t_{m,k+1})\).

Let \(z\in\mathcal{Z}\) index a prespecified treatment strategy. We
focus on two sustained strategies,
\(\mathcal{Z}=\{0,1\}\), such as never treated and always treated.
More generally, strategy \(z\) may be represented by a rule \(g_z\)
that maps the history available at each visit to the treatment required
by that strategy.

\subsubsection{Observed longitudinal data and sequential trial emulation}
\label{subsubsec:observed-data}

Suppose the observational database contains \(n\) individuals indexed
by \(i=1,\ldots,n\), and assume independence across individuals. We
adapt the longitudinal notation of \citet{su2024trialemulation} by
adding the trial index \(m\). Let \(V_i\) be a vector of
time-invariant baseline covariates.
For an individual eligible for trial \(m\), let \(L_{imk}\) denote the
time-varying covariates measured at visit \(t_{mk}\); thus, \(L_{im0}\)
contains their values at the baseline of trial \(m\). If a covariate is
not measured at every visit, \(L_{imk}\) may include its most recently
observed value together with the time elapsed since that measurement.

Let \(A_{imk}\) denote treatment during interval
\([t_{mk},t_{m,k+1})\). For binary treatment, \(A_{imk}=1\) denotes
treatment and \(A_{imk}=0\) denotes no treatment. Let \(Y_{imk}\) be
the cumulative event indicator by the beginning of interval \(k\), and
assume that the event is absorbing. That is, once the event has
occurred, the cumulative event indicator remains equal to one at all
subsequent visits. The event indicator for interval \(k\) is
\begin{equation}
  \Delta Y_{im,k+1}
  =
  Y_{im,k+1}-Y_{imk}.
  \label{eq:interval-event-increment}
\end{equation}
Among individuals who are event-free at the start of interval \(k\),
\(\Delta Y_{im,k+1}=1\) indicates that the event occurs during that
interval.
Let \(C_{im,k+1}=1\) indicate loss to follow-up after the outcome
status for interval \(k\) is recorded, and let \(C_{im,k+1}=0\)
otherwise.
We adopt the within-interval ordering
\[
  \left(
    L_{imk},
    A_{imk},
    \Delta Y_{im,k+1},
    C_{im,k+1}
  \right).
\]

An overbar denotes history within an emulated trial. For example,
\[
  \overline{L}_{imk}
  =
  (L_{im0},\ldots,L_{imk}),
  \qquad
  \overline{A}_{imk}
  =
  (A_{im0},\ldots,A_{imk}).
\]
Because \(C_{im,k+1}\) is recorded after interval \(k\), define the
censoring history available before treatment assignment at visit \(k\) as
\[
  \overline{C}_{im,1:k}
  =
  (C_{im1},\ldots,C_{imk}),
  \qquad
  \overline{C}_{im,1:0}=\varnothing.
\]
The information available immediately before treatment assignment at
visit \(k\) is
\begin{equation}
  H_{imk}
  =
  \left(
    V_i,
    \overline{L}_{imk},
    \overline{A}_{im,k-1},
    Y_{imk},
    \overline{C}_{im,1:k},
    m,
    k
  \right).
  \label{eq:observed-history}
\end{equation}
At \(k=0\), the treatment history
\(\overline{A}_{im,-1}\) is also empty.

With no grace period, an eligible person-trial is initially classified
according to the treatment received at trial entry.
For the per-protocol analysis, each person-trial contributes follow-up
until the first outcome, treatment-strategy deviation, loss to
follow-up, or administrative end of follow-up. An individual
may contribute to multiple sequential trials \citep{keogh2023sequentialtrials,su2024trialemulation}.

\subsubsection{Adherence and observed risk sets}
\label{subsubsec:risk-sets}

For strategy \(z\), define adherence through the treatment assigned at
visit \(k\) by
\begin{equation}
  D_{imk}^{z}
  =
  \prod_{\ell=0}^{k}
  I\!\left\{
    A_{im\ell}=g_z(H_{im\ell})
  \right\},
  \qquad
  D_{im,-1}^{z}=1.
  \label{eq:adherence-indicator}
\end{equation}
For a sustained static strategy, \(g_z(H_{im\ell})=z\) for every
\(\ell\). Thus, \(D_{imk}^{z}=1\) only if every observed treatment
decision from baseline through visit \(k\) agrees with strategy \(z\).
Define the event-free and uncensored indicators at the
beginning of interval \(k\) by
\begin{equation}
  \begin{gathered}
    R_{imk}
    =
    I(Y_{imk}=0),\\[4pt]
    \left\{
    \begin{aligned}
      O_{im0}
      &=
      1,
      &&\text{initial condition},\\
      O_{im,k+1}
      &=
      O_{imk}I(C_{im,k+1}=0),
      &&k=0,\ldots,\tau-1.
    \end{aligned}
    \right.
  \end{gathered}
  \label{eq:observed-risk-indicator}
\end{equation}
Define
\begin{equation}
  J_{imk}^{z}
  =
  E_{im}R_{imk}O_{imk}D_{imk}^{z}.
  \label{eq:observed-adherent-row}
\end{equation}
Thus, \(J_{imk}^{z}=1\) identifies individuals who are eligible,
event-free, uncensored at the start of interval \(k\), and adherent to
strategy \(z\) through the current treatment decision.

\subsubsection{History and state representations}
\label{subsubsec:state-representation}

The complete history \(H_{imk}\) may be too high dimensional to model
directly. We therefore summarize it by a prespecified or learned
predecision state
\begin{equation}
  S_{imk}=s(H_{imk}) \in \mathcal{S},
  \label{eq:state-representation}
\end{equation}
where \(s(\cdot)\) maps the observed history into a state space
\(\mathcal{S}\). The state may contain \(V_i\), recent values of \(L\),
and other summaries needed to predict future covariate evolution and
event risk. The state \(S_{imk}\) contains only the selected
covariate-history summaries. Trial index \(m\) and follow-up interval
\(k\) are included separately in pooled models. The current treatment
\(A_{imk}\) is not part of the state
because treatment is assigned after \(S_{imk}\) is measured.

Let \(H_{imk}^{z}\) denote the counterfactual history that would arise
if the individual followed strategy \(z\) and loss to follow-up were
prevented. Its
corresponding counterfactual state is
\begin{equation}
  S_{imk}^{z}=s(H_{imk}^{z}).
  \label{eq:counterfactual-state}
\end{equation}

\subsection{Causal Estimands}
\label{subsec:causal-estimands}

\subsubsection{Counterfactual outcomes and component risks}
\label{subsubsec:strategy-specific-risks}

To define the pooled estimand, we first introduce the counterfactual
risk within each contributing trial.
Let \(Y_{imk}^{z}\) denote the cumulative event indicator by the
beginning of interval \(k\) if individual \(i\), eligible for trial
\(m\), were to follow strategy \(z\) and remain uncensored through
follow-up. Define
\begin{equation}
  \Delta Y_{im,k+1}^{z}
  =
  Y_{im,k+1}^{z}-Y_{imk}^{z},
  \qquad
  R_{imk}^{z}
  =
  I(Y_{imk}^{z}=0).
  \label{eq:counterfactual-event-risk}
\end{equation}
For trial \(m\), the counterfactual cumulative risk by \(\tau\) is
\begin{equation}
  F_m^{z}(\tau)
  =
  \Pr\!\left(
    Y_{im,\tau}^{z}=1
    \,\middle|\,
    E_{im}=1
  \right).
  \label{eq:trial-specific-risk}
\end{equation}

\subsubsection{Pooled per-protocol estimand}
\label{subsubsec:pooling-trials}

The eligible population can differ across sequential trials because
eligibility is reassessed at the start of each trial. Sequential-trial
analyses can combine information across these trials
\citep{hernan2008observational,danaei2013observational,keogh2023sequentialtrials}.
In this paper, we define the pooled target population as the empirical
population of eligible person-trials across all trials.
This equal-person-trial target is the estimand choice adopted in this
paper, not a requirement of the MDR framework. The same framework can
also be applied to other prespecified target populations.
Giving each eligible person-trial equal baseline weight implies
\[
  N_m=\sum_{i=1}^{n}E_{im},
  \qquad
  \omega_m
  =
  \frac{N_m}
       {\sum_{m'\in\mathcal{M}}N_{m'}}.
\]
Here, \(N_m\) is the number of eligible individuals in trial \(m\), and
\(\omega_m\) is the proportion of all eligible person-trials contributed
by trial \(m\).

The pooled strategy-specific risk is
\begin{equation}
  F_{\omega}^{z}(\tau)
  =
  \sum_{m\in\mathcal{M}}
  \omega_m F_m^{z}(\tau).
  \label{eq:pooled-risk}
\end{equation}
The corresponding pooled risk difference is
\begin{equation}
  \psi_{\omega,\mathrm{RD}}(\tau)
  =
  F_{\omega}^{1}(\tau)-F_{\omega}^{0}(\tau)
  =
  \sum_{m\in\mathcal{M}}
  \omega_m
  \left\{
    F_m^{1}(\tau)-F_m^{0}(\tau)
  \right\}.
  \label{eq:pooled-risk-difference}
\end{equation}

The pooled counterfactual hazard during interval \(k\) is
\begin{equation}
  \lambda_{k}^{z,\omega}
  =
  \frac{
    \displaystyle
    \sum_{m\in\mathcal M}
    \omega_m
    \Pr\!\left(
      \Delta Y_{im,k+1}^{z}=1,\ R_{imk}^{z}=1
      \,\middle|\,
      E_{im}=1
    \right)
  }{
    \displaystyle
    \sum_{m\in\mathcal M}
    \omega_m
    \Pr\!\left(
      R_{imk}^{z}=1
      \,\middle|\,
      E_{im}=1
    \right)
  }.
  \label{eq:pooled-target-hazard}
\end{equation}
Although \(\omega_m\) defines the trial mixture at baseline, the
mixture among survivors may change over follow-up. Because the event is
absorbing, the pooled risk also satisfies
\begin{equation}
  F_{\omega}^{z}(\tau)
  =
  1-
  \prod_{k=0}^{\tau-1}
  \left(
    1-\lambda_{k}^{z,\omega}
  \right).
  \label{eq:pooled-hazard-risk}
\end{equation}

\subsection{Identification}
\label{subsec:identification}

\subsubsection{Identification assumptions}
\label{subsubsec:identification-assumptions}

Identification of the causal quantities in
Section~\ref{subsec:causal-estimands} relies on the following
conditions \citep{robins1986gformula,daniel2013timedependent}.
Formal statements and technical details are given in
Appendix~\ref{app:formal-identification-assumptions}.

\begin{enumerate}
  \item \textbf{No interference.}
  One individual's treatment does not affect another individual's
  states or outcomes.

  \item \textbf{Consistency and well-defined strategies.}
  When an individual follows strategy \(z\), the observed trajectory
  equals the trajectory under \(z\), and the strategy specifies an
  unambiguous treatment at each decision.

  \item \textbf{Sequential treatment exchangeability.}
  Conditional on the measured history, treatment decisions are not
  confounded by unmeasured predictors of future outcomes.

  \item \textbf{Sequential observation exchangeability.}
  Conditional on the measured history and treatment, remaining observed
  is not associated with future counterfactual states or outcomes.

  \item \textbf{History-level positivity.}
  Every treatment and observation pattern required by strategy \(z\)
  occurs with positive probability for the relevant measured histories.

  \item \textbf{State sufficiency and conditional-mean transportability.}
  The selected state adequately summarizes the measured history needed to predict the next event and the next state when used for target simulation. Among individuals with the same state, the next-event risk is assumed to be the same in the target and observed-adherent risk sets.

  \item \textbf{State-level overlap.}
  Every state occurring in the target risk set is represented in the
  corresponding observed-adherent risk set.

  \item \textbf{Correct temporal and risk-set alignment.}
  The target and observed data use the same within-interval ordering and
  define each risk set among individuals who are event-free at the
  beginning of that interval.
\end{enumerate}

\subsubsection{Target and observed-adherent risk-set distributions}
\label{subsubsec:target-observed-laws}

For each person-trial, define the augmented state
\begin{equation}
  \widetilde S_{imk}
  =
  (m,S_{imk}),
  \qquad
  \widetilde S_{imk}^{z}
  =
  (m,S_{imk}^{z}),
  \label{eq:augmented-pooled-state}
\end{equation}
with state space
\(\widetilde{\mathcal S}=\mathcal M\times\mathcal S\). Including
\(m\) allows the state distribution and event risk to vary across
trials within the pooled analysis.

For any measurable \(B\subseteq\widetilde{\mathcal S}\), define the
pooled target risk-set distribution under strategy \(z\) at interval
\(k\) by
\begin{equation}
  P_{k,z}^{T,\omega}(B)
  =
  \frac{
    \displaystyle
    \sum_{m\in\mathcal M}
    \omega_m
    \Pr\!\left(
      \widetilde S_{imk}^{z}\in B,\ R_{imk}^{z}=1
      \,\middle|\,
      E_{im}=1
    \right)
  }{
    \displaystyle
    \sum_{m\in\mathcal M}
    \omega_m
    \Pr\!\left(
      R_{imk}^{z}=1
      \,\middle|\,
      E_{im}=1
    \right)
  }.
  \label{eq:target-state-law}
\end{equation}
It describes the augmented states among pooled eligible person-trials
that would remain event-free under strategy \(z\) if censoring were
eliminated.

The corresponding observed-adherent distribution is
\begin{equation}
  P_{k,z}^{A,\omega}(B)
  =
  \frac{
    \displaystyle
    \sum_{m\in\mathcal M}
    \omega_m
    \Pr\!\left(
      \widetilde S_{imk}\in B,\ J_{imk}^{z}=1
      \,\middle|\,
      E_{im}=1
    \right)
  }{
    \displaystyle
    \sum_{m\in\mathcal M}
    \omega_m
    \Pr\!\left(
      J_{imk}^{z}=1
      \,\middle|\,
      E_{im}=1
    \right)
  }.
  \label{eq:adherent-state-law}
\end{equation}
Here, \(A\) means observed-adherent.

For \(\widetilde s=(m,s)\), define the observed-adherent state-specific
event risk by
\begin{equation}
  \mu_{k,z}^{A}(\widetilde s)
  =
  E\!\left(
    \Delta Y_{im,k+1}
    \,\middle|\,
    S_{imk}=s,\,
    J_{imk}^{z}=1
  \right),
  \label{eq:observed-outcome-regression}
\end{equation}
and the target state-specific event risk under strategy \(z\) by
\begin{equation}
  \mu_{k,z}^{T}(\widetilde s)
  =
  E\!\left(
    \Delta Y_{im,k+1}^{z}
    \,\middle|\,
    S_{imk}^{z}=s,\,
    E_{im}=1,\,
    R_{imk}^{z}=1
  \right).
  \label{eq:target-event-regression}
\end{equation}
The identification assumptions imply
\[
  \mu_{k}^{z}(\widetilde s)
  :=
  \mu_{k,z}^{A}(\widetilde s)
  =
  \mu_{k,z}^{T}(\widetilde s),
\]
for states in the target risk set. Thus, \(\mu_{k}^{z}\) denotes their
common identified event-risk function. Averaging this function over the
pooled target state distribution identifies the pooled target hazard:
\begin{equation}
  \lambda_{k}^{z,\omega}
  =
  \int_{\widetilde{\mathcal S}}
  \mu_{k}^{z}(\widetilde s)\,
  dP_{k,z}^{T,\omega}(\widetilde s).
  \label{eq:identified-target-hazard}
\end{equation}
Here, g-computation provides the target state distribution, while the
observed-adherent data provide the conditional event risk. The MDR
introduced below links these components by transporting the
observed-adherent distribution to the target risk set.

\subsubsection{Longitudinal g-computation}
\label{subsubsec:g-formula}

Equation~\eqref{eq:identified-target-hazard} identifies the target
hazard once \(P_{k,z}^{T,\omega}\) is known. Longitudinal g-computation
constructs this distribution by propagating pooled eligible
person-trials forward under strategy \(z\), with censoring eliminated.
Under the identification assumptions, the required event and
state-transition distributions can be learned from the observed data
\citep{robins1986gformula,daniel2013timedependent,young2011parametricgformula}.

The recursion begins with the pooled state distribution at baseline:
\begin{equation}
  P_{0,z}^{T,\omega}(B)
  =
  \frac{
    \displaystyle
    \sum_{m\in\mathcal M}
    \omega_m
    \Pr\!\left(
      \widetilde S_{im0}\in B,\ R_{im0}=1
      \,\middle|\,
      E_{im}=1
    \right)
  }{
    \displaystyle
    \sum_{m\in\mathcal M}
    \omega_m
    \Pr\!\left(
      R_{im0}=1
      \,\middle|\,
      E_{im}=1
    \right)
  }.
  \label{eq:initial-target-law}
\end{equation}
Because the baseline state is measured before the first treatment
decision, this distribution is common to both strategies.

At interval \(k\), Equation~\eqref{eq:identified-target-hazard}
averages the state-specific event risk over \(P_{k,z}^{T,\omega}\) to
obtain \(\lambda_{k}^{z,\omega}\). For person-trials that survive that
interval and for \(\widetilde s=(m,s)\), let
\begin{equation}
  \mathcal{T}_{k,z}(B\mid\widetilde s)
  =
  \Pr\!\left(
    \widetilde S_{im,k+1}^{z}\in B
    \,\middle|\,
    \substack{
      S_{imk}^{z}=s,\ E_{im}=1,\ R_{imk}^{z}=1,\\
      \Delta Y_{im,k+1}^{z}=0
    }
  \right)
  \label{eq:survivor-transition-law}
\end{equation}
denote the survivor-state transition law: the conditional probability
that the next state belongs to a measurable set
\(B\subseteq\widetilde{\mathcal S}\). The pooled target distribution at the next
interval is then
\begin{equation}
  P_{k+1,z}^{T,\omega}(B)
  =
  \frac{
    \displaystyle
    \int_{\widetilde{\mathcal S}}
    \left\{1-\mu_{k}^{z}(\widetilde s)\right\}
    \mathcal{T}_{k,z}(B\mid\widetilde s)\,
    dP_{k,z}^{T,\omega}(\widetilde s)
  }{
    1-\lambda_{k}^{z,\omega}
  },
  \label{eq:target-law-recursion}
\end{equation}
when \(1-\lambda_{k}^{z,\omega}>0\). See derivation in
Appendix~\ref{app:target-risk-set-recursion}.

\subsection{Direct Marginal Density-Ratio}
\label{subsec:direct-mdr}

For each strategy \(z\) and follow-up interval \(k\), the pooled target
and observed-adherent risk sets generally have different augmented-state
distributions. Under the state-overlap condition stated in
Appendix~\ref{app:formal-identification-assumptions}, the
Radon--Nikodym theorem guarantees the existence of the direct marginal
density ratio \citep{kallenberg2021foundations}:
\begin{equation}
  r_{k,z}^{\omega}(\widetilde s)
  =
  \frac{dP_{k,z}^{T,\omega}}
       {dP_{k,z}^{A,\omega}}(\widetilde s).
  \label{eq:direct-mdr}
\end{equation}
When both distributions have densities with respect to a common
dominating measure, this definition reduces to
\[
  r_{k,z}^{\omega}(\widetilde s)
  =
  \frac{p_{k,z}^{T,\omega}(\widetilde s)}
       {p_{k,z}^{A,\omega}(\widetilde s)}.
\]
States that are more common in the target risk set than in the
observed-adherent risk set receive larger weights. 
It is a direct marginal density ratio because it maps the
observed-adherent distribution of the current state directly to the
target distribution without using the complete longitudinal history.

\subsection{MDR-Weighted Estimation Algorithm}
\label{subsec:mdr-weighted-algorithm}

The MDR estimator has three main components: g-computation generates the
target risk sets, a classifier estimates the density ratios, and a
weighted discrete-time model estimates the interval-event hazards. The
estimand is the pooled strategy-specific risks defined in Equations~\eqref{eq:pooled-risk}.

The estimation procedure is as follows.

\begin{enumerate}
  \item \textbf{Construct the observed sequential-trial data.}
  For each trial \(m\), identify eligible individuals and create one
  record for each follow-up interval
  \citep{keogh2023sequentialtrials,su2024trialemulation}. For strategy \(z\), use records with
  \(J_{imk}^{z}=1\).

  \item \textbf{Generate the target risk-set states.}
  Fit two models using the observed longitudinal data: an event model
  for the current interval and a transition model for the next state
  among survivors. In notation, they estimate
  \begin{equation}
    \begin{aligned}
      \mu_k(a,s,m)
      &=
      \Pr\!\left(
        \Delta Y_{im,k+1}=1
        \,\middle|\,
        \substack{
          A_{imk}=a,\ S_{imk}=s,\ m,\\
          R_{imk}O_{imk}=1
        }
      \right),\\
      \mathcal{T}_k(B\mid a,s,m)
      &=
      \Pr\!\left(
        S_{im,k+1}\in B
        \,\middle|\,
        \substack{
          A_{imk}=a,\ S_{imk}=s,\ m,\\
          R_{imk}O_{imk}=1,\ \Delta Y_{im,k+1}=0,\\
          C_{im,k+1}=0
        }
      \right).
    \end{aligned}
    \label{eq:estimated-gcomp-components}
  \end{equation}
  For each strategy \(z\), draw \(N_{\mathrm{sim}}\) eligible baseline entries from
  across all sequential trials and simulate them forward by setting
  treatment to \(z\), simulating the event, removing trajectories in
  which it occurs, and updating the state among those remaining
  event-free at each interval. This is the standard
  forward-simulation procedure used in parametric longitudinal
  g-computation
  \citep{robins1986gformula,daniel2013timedependent,young2011parametricgformula}.
  Denote the resulting simulated target risk-set sample by
  \(\{\widetilde S_{uk}^{z,T}:u=1,\ldots,N_{k,z}^{T}\}\), where \(N_{k,z}^{T}\) is the number remaining event-free at interval \(k\).

  \item \textbf{Fit a follow-up-balanced ratio classifier.}
  For each strategy \(z\) and follow-up interval \(k\), sample target
  rows so that their number equals the number of observed-adherent
  rows. Label target rows by \(G=1\) and observed-adherent rows by
  \(G=0\). The classifier estimates the conditional probability that
  an augmented state comes from the target risk
  set. Pool
  the balanced samples across follow-up intervals and trials, and
  fit one probabilistic classifier for each strategy:
  \begin{equation}
    h_z(k,\widetilde s)
    =
    \Pr\!\left(
      G=1
      \,\middle|\,
      k,\widetilde S=\widetilde s,z
    \right).
    \label{eq:mdr-classifier-probability}
  \end{equation}

  \item \textbf{Convert classifier probabilities to normalized MDR
  weights.}
  Let \(\widehat h_z(k,\widetilde s)\) be the fitted probability and
  let \(\widehat\rho_z\) be the proportion of target rows in the
  classifier training sample. By
  Lemma~\ref{lem:classifier-density-ratio}, the raw density-ratio
  estimate is
  \begin{equation}
    \widehat r_{k,z}^{\mathrm{raw}}(\widetilde s)
    =
    \frac{1-\widehat\rho_z}{\widehat\rho_z}
    \frac{\widehat h_z(k,\widetilde s)}
         {1-\widehat h_z(k,\widetilde s)}.
    \label{eq:classifier-mdr}
  \end{equation}
  Classifier odds are a standard approach to density-ratio estimation
  \citep{menon2016densityratio,bickel2009covariateshift}. With exact
  follow-up balancing, the target and observed-adherent class
  proportions are equal within each interval and in the combined
  training sample, so \(\widehat\rho_z=1/2\).

  To enforce the mean-one property of a density ratio in each
  \((z,k)\), normalize the raw weights as
  \begin{equation}
    \widehat r_{imk}^{z}
    =
    \frac{
      \widehat r_{k,z}^{\mathrm{raw}}
      (\widetilde S_{imk})
    }{
      \displaystyle
      \frac{1}{n_{k,z}^{A}}
      \sum_{i',m':\,J_{i'm'k}^{z}=1}
      \widehat r_{k,z}^{\mathrm{raw}}
      (\widetilde S_{i'm'k})
    },
    \qquad
    n_{k,z}^{A}
    =
    \sum_{i,m}J_{imk}^{z}.
    \label{eq:cell-normalized-mdr}
  \end{equation}

  \item \textbf{Fit the MDR-weighted interval-event model.}
  Among observed-adherent rows, fit a pooled discrete-time outcome
  model \(\mu_{\beta_r}(z,m,k,s)\) using
  \(\widehat r_{imk}^{z}\) as the observation weight. For a binary
  outcome, let
  \(\ell_{\mathrm{Bern}}\) denote the Bernoulli log-likelihood
  contribution. The coefficient estimator maximizes
  \begin{equation}
    \widehat\beta_r
    =
    \arg\max_{\beta_r}
    \sum_{i,m,z,k}
    J_{imk}^{z}
    \widehat r_{imk}^{z}
    \ell_{\mathrm{Bern}}
    \left\{
      \Delta Y_{im,k+1},
      \mu_{\beta_r}(z,m,k,S_{imk})
    \right\}.
    \label{eq:mdr-weighted-outcome-fit}
  \end{equation}

  \item \textbf{Standardize the fitted hazards to the target states.}
  Evaluate the fitted outcome model on every target-state
  simulation and average within each cell \((z,k)\):
  \begin{equation}
    \widehat\lambda_{k}^{z,\omega}
    =
    \frac{1}{N_{k,z}^{T}}
    \sum_{u=1}^{N_{k,z}^{T}}
    \mu_{\widehat\beta_r}
    \left(
      z,m_u,k,S_{uk}^{z,T}
    \right).
    \label{eq:mdr-standardized-hazard}
  \end{equation}
  The corresponding pooled cumulative risk and risk difference are
  \begin{equation}
    \widehat F_{\omega}^{z}(\tau)
    =
    1-
    \prod_{k=0}^{\tau-1}
    \left(
      1-\widehat\lambda_{k}^{z,\omega}
    \right),
    \qquad
    \widehat\psi_{\omega,\mathrm{MDR}}(\tau)
    =
    \widehat F_{\omega}^{1}(\tau)
    -
    \widehat F_{\omega}^{0}(\tau).
    \label{eq:mdr-weighted-risk-difference}
  \end{equation}

\end{enumerate}

\subsection{Doubly Robust STTE--MDR Estimation}
\label{subsec:dr-stte-mdr}

The doubly robust estimator has two parts: an outcome-regression
estimate standardized to the target states and an MDR-weighted
correction for its prediction error
\citep{bang2005doublyrobust,dahabreh2020extending,kallus2020drl}. It uses the simulated target
states and MDR weights from
Section~\ref{subsec:mdr-weighted-algorithm}. For this estimator,
the pooled outcome model is fitted without MDR weights. Let
\(\widehat\beta\) denote its fitted coefficient vector, and let
\[
  \widehat\mu_{imk}^{z}
  =
  \mu_{\widehat\beta}(z,m,k,S_{imk})
\]
denote its predicted event probability.

For each \((z,k)\) cell, average the outcome-model predictions over the
simulated target states. This gives the outcome-regression component.
Next, calculate each observed residual as the observed outcome minus its
prediction. Use the MDR weights to transport the average residual from
the observed-adherent risk set to the target risk set. Adding this
correction gives
\begin{equation}
  \begin{aligned}
    \widehat\lambda_{k}^{z,\omega,\mathrm{plug}}
    &=
    \frac{1}{N_{k,z}^{T}}
    \sum_{u=1}^{N_{k,z}^{T}}
    \mu_{\widehat\beta}
    \left(z,m_u,k,S_{uk}^{z,T}\right),\\
    \widehat c_{k}^{z,\omega}
    &=
    \frac{
      \displaystyle
      \sum_{i,m}
      J_{imk}^{z}\widehat r_{imk}^{z}
      \left(
        \Delta Y_{im,k+1}-\widehat\mu_{imk}^{z}
      \right)
    }{
      \displaystyle
      \sum_{i,m}
      J_{imk}^{z}\widehat r_{imk}^{z}
    },\\
    \widehat\lambda_{k}^{z,\omega,\mathrm{DR}}
    &=
    \Pi_{[0,1]}
    \left(
      \widehat\lambda_{k}^{z,\omega,\mathrm{plug}}
      +
      \widehat c_{k}^{z,\omega}
    \right),
  \end{aligned}
  \label{eq:dr-mdr-hazard}
\end{equation}
where \(\Pi_{[0,1]}(x)=\min\{1,\max(0,x)\}\). The projection keeps the corrected
hazard within its allowable range and serves as a finite-sample
safeguard. The cumulative risks and risk difference are
\begin{equation}
  \widehat F_{\omega,\mathrm{DR}}^{z}(\tau)
  =
  1-
  \prod_{k=0}^{\tau-1}
  \left(
    1-\widehat\lambda_{k}^{z,\omega,\mathrm{DR}}
  \right),
  \qquad
  \widehat\psi_{\omega,\mathrm{DR}}(\tau)
  =
  \widehat F_{\omega,\mathrm{DR}}^{1}(\tau)
  -
  \widehat F_{\omega,\mathrm{DR}}^{0}(\tau).
  \label{eq:dr-mdr-risk-difference}
\end{equation}

The formal statement and the qualification concerning
target-state generation appear in Theorem~\ref{thm:dr-mdr}.

\subsection{Theoretical Properties}
\label{subsec:theoretical-properties}

This section summarizes the main theoretical guarantees; proofs and
technical qualifications are in
Appendix~\ref{app:theoretical-proofs}. Fix a supported pooled
strategy--follow-up cell \(c=(z,k)\). To suppress indices in the
following cellwise results, write
\(P_c^T=P_{k,z}^{T,\omega}\),
\(P_c^A=P_{k,z}^{A,\omega}\), and
\(r_c=dP_c^T/dP_c^A\). Let \(S_c\) denote a generic augmented state
with these distributions, and let \(Y_c\) denote its corresponding
interval-event indicator under the observed-adherent joint distribution.
Define
\(\mu_c(s)=E(Y_c\mid S_c=s)\) as the observed-adherent conditional event
risk. Expectations involving \(Y_c\) use the joint observed-adherent
law, whose state marginal is \(P_c^A\). The target hazard is
\(\lambda_c=\lambda_k^{z,\omega}
=E_{P_c^T}\{\mu_c(S_c)\}\).

The results proceed from the distributional foundation of MDR to its
causal and estimation implications. We first establish the
change-of-measure identity, then show identification and classifier
recovery, compare the direct MDR with full-history weighting, and
finally state the conditional double-robustness result.

\subsubsection{Change-of-measure representation}
\label{subsubsec:mdr-change-measure}

Weighting by the density ratio transforms expectations under the
observed-adherent distribution into the corresponding expectations
under the target distribution.
Proposition~\ref{prop:mdr-change-measure} states this change-of-measure
identity formally and gives its normalization and set-balance
consequences.

\begin{proposition}[Change of measure and balance]
\label{prop:mdr-change-measure}
If \(P_c^T\ll P_c^A\), then for every \(P_c^T\)-integrable function
\(\phi\),
\begin{equation}
  E_{P_c^T}\{\phi(S_c)\}
  =
  E_{P_c^A}\{r_c(S_c)\phi(S_c)\}.
  \label{eq:theory-change-measure}
\end{equation}
Consequently, \(r_c\) is normalized under the observed-adherent
distribution:
\[
  E_{P_c^A}\{r_c(S_c)\}=1.
\]
Moreover, for every measurable
\(B\subseteq\widetilde{\mathcal S}\),
\[
  E_{P_c^A}
  \{r_c(S_c)I(S_c\in B)\}
  =
  P_c^T(B).
\]
\end{proposition}

\subsubsection{Identification by marginal density-ratio weighting}
\label{subsubsec:mdr-identification}

Under the assumptions in
Section~\ref{subsubsec:identification-assumptions}, applying
Proposition~\ref{prop:mdr-change-measure} with
\(\phi(s)=\mu_c(s)\) and using
iterated expectation yields the
MDR representation of the target hazard in
Theorem~\ref{thm:mdr-identification}.

\begin{theorem}[Identification by direct MDR]
\label{thm:mdr-identification}
Under the identification conditions in
Section~\ref{subsubsec:identification-assumptions}, for every supported
strategy--follow-up cell,
\begin{equation}
  \lambda_c
  =
  E_{P_c^A}
  \left\{
    r_c(S_c)Y_c
  \right\}.
  \label{eq:theory-mdr-identification}
\end{equation}
The identified hazards determine the pooled cumulative risks and risk
difference through Equations~\eqref{eq:pooled-hazard-risk}
and~\eqref{eq:pooled-risk-difference}.
\end{theorem}

Theorem~\ref{thm:mdr-identification} assumes the true density ratio.
The following lemma justifies the classifier-odds representation used in
Equation~\eqref{eq:classifier-mdr}.

\begin{lemma}[Classifier recovery of the density ratio]
\label{lem:classifier-density-ratio}
Construct a classifier sample within cell \(c\) by drawing
\(S_c\sim P_c^T\) when \(G=1\) and \(S_c\sim P_c^A\) when \(G=0\).
Let \(\rho_c=\Pr_{\mathrm{tr}}(G=1\mid c)\) and
\(h_c(s)=\Pr_{\mathrm{tr}}(G=1\mid S_c=s,c)\). If
\(0<\rho_c<1\) and \(P_c^T\ll P_c^A\), then
\begin{equation}
  r_c(s)
  =
  \frac{1-\rho_c}{\rho_c}
  \frac{h_c(s)}{1-h_c(s)}
  \qquad P_c^A\text{-almost surely}.
  \label{eq:theory-classifier-ratio}
\end{equation}
For balanced target and observed-adherent classes, \(\rho_c=1/2\), and
the identity reduces to
\[
  r_c(s)=\frac{h_c(s)}{1-h_c(s)}.
\]
\end{lemma}

\subsubsection{Direct MDR versus cumulative full-history weighting}
\label{subsubsec:mdr-cumulative-weights}

Conventional longitudinal weights accumulate interval-specific factors
for treatment adherence and observation through interval \(k\)
\citep{robins2000msm,keogh2023sequentialtrials}.

For the fixed cell \(c\), let
\(H_c=\widetilde H_k=(m,H_{imk})\) denote a generic augmented history,
and let \(Q_c^T\) and \(Q_c^A\) denote its aligned target and
observed-adherent distributions, respectively. When
\(Q_c^T\ll Q_c^A\), define
\begin{equation}
  W_c(h)
  =
  \frac{dQ_c^T}
       {dQ_c^A}(h).
  \label{eq:full-history-ratio}
\end{equation}
Let \(w_c^{\mathrm{IPACW}}(h)\) denote the correctly specified raw
cumulative weight based on inverse probabilities of adherence and
censoring through interval \(k\), corresponding to 
Equation~\eqref{eq:unstabilized_weight}. Define

\begin{equation}
  \widetilde w_c^{\mathrm{IPACW}}(h)
  :=
  \frac{
    w_c^{\mathrm{IPACW}}(h)
  }{
    E_{Q_c^A}\{w_c^{\mathrm{IPACW}}(H_c)\}
  }.
  \label{eq:normalized-ipacw-history-ratio}
\end{equation}

Theorem~\ref{thm:mdr-rao-blackwell} shows that the true population MDR is the
conditional expectation of the full-history ratio given the augmented state
\(S_c\). This is a Rao--Blackwellization
result
\citep{rao1945information,blackwell1947conditional,liu2020conditional}
and is related to marginalized importance sampling
\citep{liu2018breakinghorizon,xie2019mis}.

\begin{theorem}[MDR as a marginalized trajectory weight]
\label{thm:mdr-rao-blackwell}
Let \(Q_c^T\) and \(Q_c^A\) be aligned full-history laws whose
pushforwards under \(S_c=s(H_c)\) are \(P_c^T\) and \(P_c^A\).
If \(Q_c^T\ll Q_c^A\) and \(W_c=dQ_c^T/dQ_c^A\), then
\begin{equation}
  r_c(S_c)
  =
  E_{Q_c^A}(W_c\mid S_c)
  \qquad Q_c^A\text{-almost surely}.
  \label{eq:theory-mdr-marginalization}
\end{equation}
If \(W_c\in L_2(Q_c^A)\), then
\begin{equation}
  \operatorname{Var}_{P_c^A}\{r_c(S_c)\}
  \leq
  \operatorname{Var}_{Q_c^A}(W_c).
  \label{eq:theory-mdr-variance}
\end{equation}
The variance inequality compares the true population density ratios and
does not imply a variance ordering for estimators based on estimated
density ratios.
\end{theorem}

For the aligned full-history change of measure, we show 
\(\widetilde w_c^{\mathrm{IPACW}}(h)=W_c(h)\). Substitution into
Equation~\eqref{eq:theory-mdr-marginalization} gives
\begin{equation}
  \begin{aligned}
  r_c(S_c)
  &=E_{Q_c^A}
    \{\widetilde w_c^{\mathrm{IPACW}}(H_c)\mid S_c\}\\
  &=\frac{
    E_{Q_c^A}\{w_c^{\mathrm{IPACW}}(H_c)\mid S_c\}
  }{
    E_{Q_c^A}\{w_c^{\mathrm{IPACW}}(H_c)\}
  },
  \qquad Q_c^A\text{-almost surely}.
  \end{aligned}
  \label{eq:mdr-ipacw-marginalization}
\end{equation}
Thus, the true population MDR is the conditional marginalization of the mean-normalized cumulative IPACW over histories sharing the same augmented state \(S_c\). See Appendix~\ref{app:rao-blackwell} for details.

\subsubsection{Conditional double robustness}
\label{subsubsec:dr-mdr-theory}

The estimator in Section~\ref{subsec:dr-stte-mdr} augments the MDR weighting component
with an outcome regression \citep{kallus2020drl}. The following theorem
states its conditional double-robustness property and makes explicit
the requirement that the target-state distribution be generated
consistently.

\begin{theorem}[Conditional double robustness of normalized DR--MDR]
\label{thm:dr-mdr}
Fix a cell \(c\), and let \(\widehat\Psi_c\) denote the unprojected
DR--MDR estimator in Equation~\eqref{eq:dr-mdr-hazard}, and
\(\widehat\lambda_c^{\mathrm{DR}}
=\Pi_{[0,1]}(\widehat\Psi_c)\). Suppose that
the target-state law is consistently estimated. Let
\(\overline\mu_c\) and \(\overline r_c\) denote the probability limits
of the fitted outcome regression and density ratio respectively, so
that
\[
  \widehat\mu_c\xrightarrow{p}\overline\mu_c,
  \qquad
  \widehat r_c\xrightarrow{p}\overline r_c,
\]
Assume \(0<E_{P_c^A}\{\overline r_c(S_c)\}<\infty\), and suppose that the
corresponding empirical averages converge to their population
expectations. If either
\[
  \overline\mu_c=\mu_c
  \qquad\text{or}\qquad
  \overline r_c=a_c r_c
  \quad\text{for some }a_c>0,
\]
then
\[
  \widehat\Psi_c
  \xrightarrow{p}
  \lambda_c,
  \qquad
  \widehat\lambda_c^{\mathrm{DR}}
  \xrightarrow{p}
  \lambda_c.
\]
\end{theorem}

%% file: simulation_study.tex

\section{Simulation Study}
\label{sec:simulation-study}

\subsection{Simulation Design and Data-Generating Process}
\label{subsec:simulation-dgp}

We specified a data-generating process (DGP) for a longitudinal cohort
with 3-dim time-invariant covariates
\(X_i\), a continuous time-varying confounder \(L_{it}\), treatment
\(A_{it}\), the interval event indicator \(\Delta Y_{i,t+1}\), and the
interval loss-to-follow-up indicator \(\Delta C_{i,t+1}\). The full cohort observation contained
\(K=24\) monthly follow-up intervals, indexed by
\(t=0,\ldots,K-1\).  Within each month, variables were generated in the order
\[
  L_{it}
  \longrightarrow A_{it}
  \longrightarrow \Delta Y_{i,t+1}
  \longrightarrow \Delta C_{i,t+1}
  \longrightarrow L_{i,t+1}.
\]
The current covariate value \(L_{it}\) influenced both treatment
assignment \(A_{it}\) and the subsequent event risk. In turn,
\(A_{it}\) influenced the next covariate value \(L_{i,t+1}\). 

Table~\ref{tab:simulation-reference-dgp} summarizes the variables included
in each reference data-generating model. Complete baseline distributions,
model coefficients, conditional sampling details, and probability restrictions are provided in
Appendix~\ref{app:simulation-design}. 

\begin{table}[htbp]
  \centering
  \small
  \caption{Reference longitudinal model specifications.}
  \begin{tabular}{%
    >{\raggedright\arraybackslash}p{0.32\textwidth}
    >{\raggedright\arraybackslash}p{0.60\textwidth}}
    \hline
    Component & Reference specification \\
    \hline
    Baseline
    & Time-invariant covariates \(X_i=(X_{i1},X_{i2},X_{i3})\)
      and time-varying covariate \(L_{i0}\). \\[2pt]
    Treatment (logistic)
    & \(A_{it}\sim A_{i,t-1}+L_{it}+X_i+t\). \\[2pt]
    Event (logistic)
    & \(\Delta Y_{i,t+1}\sim A_{it}+L_{it}+X_i+t\). \\[2pt]
    Loss to follow-up (logistic)
    & \(\Delta C_{i,t+1}\sim A_{it}+L_{it}+X_i\). \\[2pt]
    Evolution of \(L\) (linear)
    & \(L_{i,t+1}\sim L_{it}+A_{it}+X_i\). \\
    \hline
  \end{tabular}
  \label{tab:simulation-reference-dgp}
\end{table}

From each cohort, we emulated 12-month (\(\tau=12\)) target trials beginning at months
\(m=0,\ldots,12\). Individuals were eligible at a given trial entry if
they were event-free, uncensored, and untreated immediately beforehand.
They could enter multiple trials as long as they continued to meet these
criteria.

We compared two sustained strategies: initiate treatment at entry and remain treated (\(g_1\)), or remain untreated throughout
follow-up (\(g_0\)). Observed follow-up ended at the first event,
protocol deviation, loss to follow-up, or the 12-month administrative
endpoint.

For each scenario, the true causal contrast, pooled 12-month risk difference, was approximated using a separate Monte Carlo population of \(50{,}000\) individuals \citep{keogh2023sequentialtrials}. 
Further details of this calculation are given in
Appendix~\ref{app:simulation-design}. The implementation and simulation analysis code is available at
\url{https://github.com/Zernjk/STTE-MDR}.

\subsection{Simulation Scenarios}
\label{subsec:simulation-scenarios}

We considered a base scenario and six stress tests, summarized in
Table~\ref{tab:simulation-scenarios}. Each stress test changed one
component of the reference data-generating process while holding the
others fixed. Exact coefficient changes are reported in
Appendix~\ref{app:scenario-diagnostics}.

\begin{table}[htbp]
  \centering
  \small
  \caption{Simulation scenarios and intended challenges.}
  \begin{tabular}{%
    >{\raggedright\arraybackslash}p{0.17\textwidth}
    >{\raggedright\arraybackslash}p{0.40\textwidth}
    >{\raggedright\arraybackslash}p{0.30\textwidth}}
    \hline
    Scenario & Change from the base scenario & Intended challenge \\
    \hline
    Base
    & \(n=2000\) and reference model parameters
    & Standard setting \\[2pt]
    Small sample
    & Reduce the cohort size to \(n=100\)
    & Finite-sample variability \\[2pt]
    Strong confounding
    & Increase the coefficient of \(L_{it}\) in the treatment and event models
    & Stronger time-varying confounding \\[2pt]
    Poor positivity
    & Increase the coefficient of \(L_{it}\) in the treatment model
    & Practical nonoverlap \\[2pt]
    Poor adherence
    & Decrease the coefficient of \(A_{i,t-1}\) in the treatment model
    & Fewer adherent observations \\[2pt]
    Rare event
    & Decrease the event-model intercept
    & Sparse events \\[2pt]
    Nonlinear
    & Add \(L_{it}^{2}\) to the treatment model, and add
      \(L_{it}^{2}\) and \(A_{it}L_{it}\) to the event model
    & Nonlinearity and model misspecification \\
    \hline
  \end{tabular}
  \label{tab:simulation-scenarios}
\end{table}

Figure~\ref{fig:main-scenario-diagnostics} confirms that the scenarios
produced the intended contrasts. Strong confounding approximately doubled
the two confounding-association contrasts. Poor positivity increased the
percentage of extreme treatment probabilities from \(10.9\%\) to
\(40.3\%\), poor adherence reduced treatment continuation from \(76.3\%\)
to \(55.5\%\), and the rare-event setting reduced the 12-month event risk
from \(12.5\%\) to \(1.1\%\). The small-sample setting mainly increased the
variation between replication. Diagnostic definitions,
numerical summaries, and supporting plots are provided in
Appendix~\ref{app:scenario-diagnostics}.

\begin{figure}[!htbp]
  \centering
  \includegraphics[width=\textwidth]{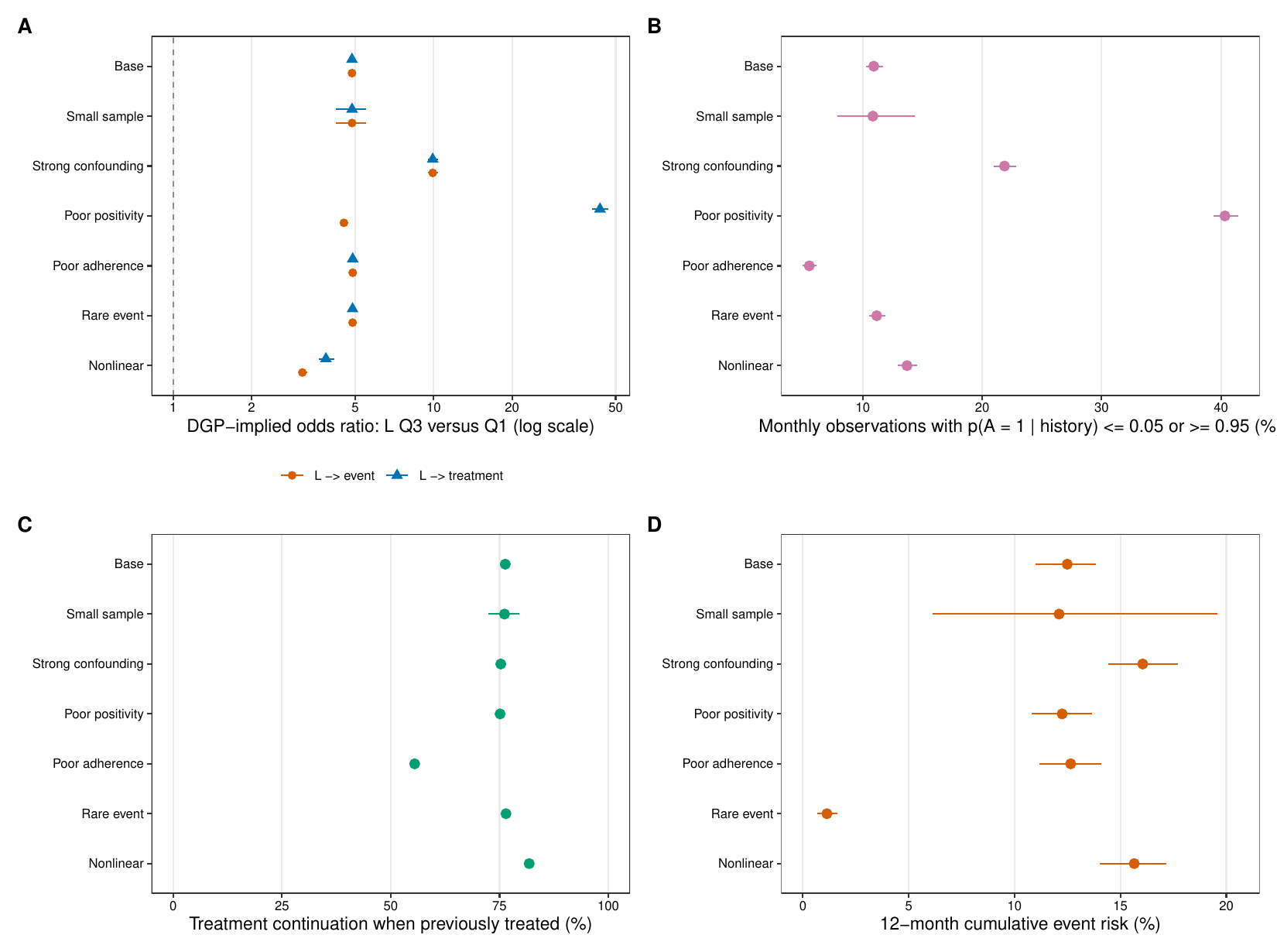}
  \caption{Empirical characteristics of the seven simulation scenarios
  across 500 replications per scenario. Points denote medians and
  horizontal intervals denote empirical 2.5th and 97.5th percentiles
  across replications. The panels summarize confounding strength,
  practical positivity, treatment continuation, and 12-month event risk,
  respectively. The intervals describe between-replication variation.
  DGP denotes the data-generating process. In Panel A, we compared the odds
  of treatment assignment and the odds of an event when \(L_{it}\) was set
  to its third rather than first quartile, while all other variables were
  unchanged. Larger odds ratios indicate stronger associations of \(L_{it}\)
  with treatment and event risk and therefore stronger confounding.}
  \label{fig:main-scenario-diagnostics}
\end{figure} 

\subsection{Estimators Compared}
\label{subsec:estimators-compared}

We compared ten estimators of the pooled 12-month per-protocol risk
difference in Equation~\eqref{eq:pooled-risk-difference}.
\begin{itemize}
  \item \textbf{Marginal density ratio (MDR).}\par
  MDR--\emph{Entry}, MDR--\emph{State}, \emph{DR}-MDR--\emph{Entry}, and \emph{DR}-MDR--\emph{State}.

  \item \textbf{Inverse probability of censoring weighting (IPCW).}\par
  IPCW--\emph{Entry}, IPCW--\emph{State}, \emph{DR}-IPCW--\emph{Entry}, and \emph{DR}-IPCW--\emph{State}.

  \item \textbf{TrialEmulation.}\par
  We used the \texttt{TrialEmulation} R package
  \citep{su2024trialemulation,gravestock2026trialemulation}, which combined a
  weighted pooled-logistic outcome model with stabilized weights for treatment
  switching and loss to follow-up. Compared with IPCW--\emph{Entry}, TrialEmulation
  omitted treatment reweighting at trial entry and shifted the
  natural-censoring weight forward by one interval.

  \item \textbf{Longitudinal TMLE--GLM.}\par
  This estimator used generalized linear models for the sequential outcome,
  treatment, and censoring components of longitudinal targeted maximum
  likelihood estimation. It is doubly robust because it combines sequential
  outcome regression with a targeting step based on the treatment and
  censoring mechanisms \citep{lendle2017ltmle}.
\end{itemize}

The suffix \emph{Entry} denotes an outcome model based on trial-entry
covariates, followed by standardization over the eligible trial-entry
population. \emph{State} denotes an outcome model using time-updated covariates, with interval hazards standardized over the corresponding simulated target risk set within each \((z,k)\). The prefix \emph{DR} indicates that
the corresponding plug-in estimator was augmented by a weighted
outcome-residual correction. We focus below on the IPCW and MDR
estimators; complete specifications for other implementation choices are provided in
Appendix~\ref{app:simulation-estimator-details}.

\subsubsection{Target-state simulation}

The IPCW--\emph{State} and \emph{DR}-IPCW--\emph{State} estimators, as well as all four MDR
estimators, used simulated strategy-specific target risk sets. For person
\(i\), trial \(m\), and follow-up month \(k\), the state in
Section~\ref{subsubsec:state-representation} was specified as
\begin{equation}
  S_{imk}
  =
  \left(
    X_i,\,
    L_{i0},\,
    L_{im,k-1},\,
    L_{imk}
  \right),
  \label{eq:simulation-mdr-state}
\end{equation}
where \(L_{im,k-1}\) and \(L_{imk}\) are the previous and current values
of the time-varying covariate.

The transition and event
models on the original cohort time scale \(t\) were
\begin{align}
  \text{Covariate transition:}\qquad
  L_{i,t+1}
  &\sim
  X_i+L_{i0}+L_{i,t-1}+L_{it}+A_{it}+t,
  \label{eq:simulation-mdr-transition-model}\\
  \text{Event depletion:}\qquad
  \Delta Y_{i,t+1}
  &\sim
  X_i+L_{i0}+L_{i,t-1}+L_{it}+A_{it}+t.
  \label{eq:simulation-mdr-target-event-model}
\end{align}
Fitting and forward-simulation details are given in
Appendix~\ref{app:simulation-estimator-details}.

\subsubsection{Weight construction}

\paragraph{MDR weights.}
Let \(G=1\) identify a simulated target state and \(G=0\) an
observed-adherent state. For each strategy, we fitted a 500-tree
probability forest pooled across follow-up months
\citep{breiman2001randomforests}. It estimated
\begin{equation}
  h_z
  =
  \Pr\!\left[
    G=1
    \,\middle|\,
    k,m,S_{imk},
    L_{imk}-L_{im,k-1},
    L_{imk}^2,
    L_{imk}^3
  \right],
  \label{eq:simulation-mdr-classifier}
\end{equation}
with \(k\) represented by indicators. We used a nonlinear classifier to
allow the target and observed-adherent state distributions to differ
through nonlinearities and interactions. 

\paragraph{IPCW weights.}
The IPCW estimators used stabilized inverse probabilities for protocol
deviation and natural loss to follow-up. With no grace period, the
stabilized artificial-censoring weight for strategy \(z\) was
\begin{align}
  \widehat{\mathrm{SW}}_{imk}^{A}(z)
  ={}&
  \prod_{j=0}^{k}
  \frac{
    \widehat{\Pr}\!\left(
      A_{imj}=z
      \mid j,m,X_i,A_{im,j-1}
    \right)
  }{
    \widehat{\Pr}\!\left(
      A_{imj}=z
      \mid j,m,X_i,L_{imj},A_{im,j-1}
    \right)
  }.
  \label{eq:simulation-stabilized-artificial-weight}
\end{align}
The corresponding natural-censoring weight was
\begin{align}
  \widehat{\mathrm{SW}}_{imk}^{C}(z)
  ={}&
  \prod_{j=0}^{k-1}
  \frac{
    \widehat{\Pr}\!\left(
      \Delta C_{im,j+1}=0
      \,\middle|\,
      \substack{
        J_{imj}^{z}=1,\ \Delta Y_{im,j+1}=0,\\
        z,j,m,X_i
      }
    \right)
  }{
    \widehat{\Pr}\!\left(
      \Delta C_{im,j+1}=0
      \,\middle|\,
      \substack{
        J_{imj}^{z}=1,\ \Delta Y_{im,j+1}=0,\\
        j,m,X_i,L_{imj},A_{imj}
      }
    \right)
  }.
  \label{eq:simulation-stabilized-natural-weight}
\end{align}
The combined weight was
\begin{equation}
  \widehat{\mathrm{SW}}_{imk}(z)
  =
  \widehat{\mathrm{SW}}_{imk}^{A}(z)
  \widehat{\mathrm{SW}}_{imk}^{C}(z).
  \label{eq:simulation-stabilized-ipcw}
\end{equation}
The numerator and denominator probabilities were estimated by logistic
regression. Risk-set restrictions, product alignment, probability bounds,
and truncation choices are reported in
Appendix~\ref{app:simulation-estimator-details}.

\subsubsection{Outcome models}

The \emph{Entry} and \emph{State} specifications used different covariate
information in the outcome regression. They also standardized risks over
different populations. Their logistic outcome models were
\begin{align}
  \emph{Entry}:\qquad
  \Delta Y_{im,k+1}
  &\sim
  z+\operatorname{factor}(k)+m+X_i+L_{im0},
  \label{eq:simulation-entry-outcome-model}\\
  \emph{State:}\qquad
  \Delta Y_{im,k+1}
  &\sim
  z+\operatorname{factor}(k)+m+X_i+L_{i0}+L_{imk}.
  \label{eq:simulation-state-outcome-model}
\end{align}
Here, \(\operatorname{factor}(k)\) denotes categorical indicators for
follow-up month. IPCW and MDR
plug-in estimators fitted the appropriate model using their respective
weights. The DR estimators fitted the same model without outcome weights
and applied the weights in the residual correction described in
Section~\ref{subsec:dr-stte-mdr}.

\subsection{Simulation Results}
\label{subsec:simulation-results}

Let \(\psi^{(s)}\) denote the Monte Carlo truth under scenario \(s\), and
let \(\widehat\psi_{r}^{(s)}\) denote its estimate in replication \(r\),
where \(r=1,\ldots,N_{\mathrm{rep}}\). With
\(N_{\mathrm{rep}}=500\) replications per scenario, signed bias and root
mean squared error (RMSE) were calculated as
\begin{align}
  \operatorname{Bias}^{(s)}
  &=
  \frac{1}{N_{\mathrm{rep}}}
  \sum_{r=1}^{N_{\mathrm{rep}}}
  \left(\widehat\psi_{r}^{(s)}-\psi^{(s)}\right),
  \label{eq:simulation-bias}\\
  \operatorname{RMSE}^{(s)}
  &=
  \left\{
    \frac{1}{N_{\mathrm{rep}}}
    \sum_{r=1}^{N_{\mathrm{rep}}}
    \left(\widehat\psi_{r}^{(s)}-\psi^{(s)}\right)^2
  \right\}^{1/2}.
  \label{eq:simulation-rmse}
\end{align}
Both metrics refer to the estimated pooled 12-month risk difference and
are reported in percentage points. For methods using explicit analysis
weights, stability was assessed using the replication-specific maximum
weight and coefficient of variation.

As shown in Figure~\ref{fig:simulation-performance}, an MDR or \emph{DR}-MDR
estimator attained the lowest RMSE in six of the seven scenarios;
\emph{DR}-IPCW--\emph{Entry} was best in the small-sample setting.
\emph{State}-based MDR estimators had the lowest RMSE in the base,
strong-confounding, poor-positivity, poor-adherence, and nonlinear
scenarios, whereas \emph{DR}-MDR--\emph{Entry} had the lowest RMSE with rare events.
\emph{DR}-MDR--\emph{State} also had the smallest absolute bias in five scenarios.
The clearest advantage occurred under poor positivity, where MDR--\emph{State}
had an RMSE of \(3.48\) percentage points, compared with \(7.96\),
\(9.54\), and \(7.85\) points for IPCW--\emph{Entry}, IPCW--\emph{State}, and
TrialEmulation, respectively. Nonlinearity remained the most challenging
setting, although \emph{DR}-MDR--\emph{State} achieved the lowest RMSE
(\(8.71\) points).

\begin{figure}[htbp]
  \centering
  \includegraphics[width=0.9\linewidth]{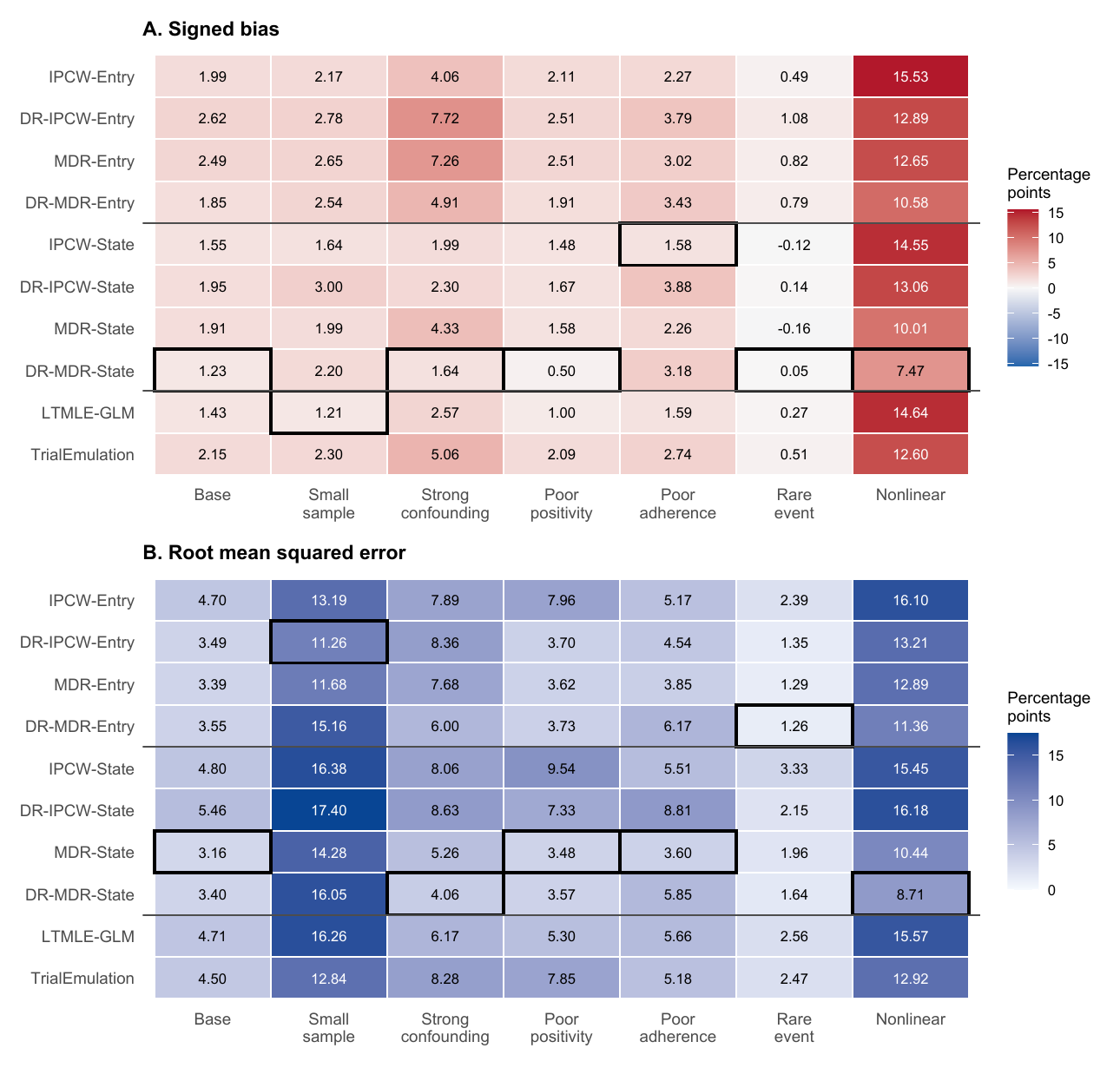}
  \caption{Monte Carlo performance of the ten estimators across seven
  scenarios, based on 500 replications per scenario. Entries are
  percentage points. Panel~A shows signed bias, with positive values
  indicating upward bias in the estimated risk difference. Panel~B
  shows root mean squared error. Outlined cells mark the smallest
  absolute bias or RMSE within each scenario. Horizontal rules separate
  Entry estimators, State estimators, and the two reference
  implementations.}
  \label{fig:simulation-performance}
\end{figure}

MDR generally performed well because it transported observed-adherent states directly to the target distribution within each \((z, k)\) cell. By avoiding cumulative probability products, MDR produced more stable weights, which likely contributed to its lower RMSE. \emph{State}-based standardization may reduce bias by allowing the covariate distribution to evolve over follow-up. The \emph{DR} correction further reduce bias from outcome-model misspecification, although the remaining bias of \emph{DR}-MDR--\emph{State} in the nonlinear scenario shows that augmentation did not fully overcome misspecification. In the small-sample setting, however, sparse adherent risk sets made the target-state, density-ratio, and outcome models less precise. This weakened the MDR advantage, although MDR–Entry remained close to the best-performing estimator.

This interpretation was consistent with the weight diagnostics: MDR
weights were substantially more stable than IPCW and TrialEmulation
weights, particularly under poor positivity. Complete results are provided
in Appendix~\ref{app:simulation-weight-diagnostics}.

%% file: discussion.tex

\section{Discussion}
\label{sec:discussion}

This study developed STTE--MDR, a marginal density-ratio approach for
estimating pooled per-protocol effects in sequential target trial
emulations \citep{hernan2016targettrial,keogh2023sequentialtrials,
su2024trialemulation}. The method treats estimation as a risk-set
transport problem.

The central innovation is a direct density ratio comparing the target
risk set with the observed-adherent risk set. Conventional longitudinal
weighting forms cumulative weights from a sequence of treatment-adherence
and censoring probabilities. A few small probabilities can make the
resulting weights unstable \citep{robins2000msm,cole2008constructing}.
MDR instead
jointly accounts for how nonadherence and censoring alter the current
state distribution. 
A probabilistic classifier provides a practical way to estimate this ratio \citep{menon2016densityratio,bickel2009covariateshift}. 
Under the stated identification conditions, MDR weighting recovers the target
interval hazards. At the population level, the MDR weights have no
greater variance than the corresponding full-history weights. However,
less variable weights do not necessarily produce less variable
treatment-effect estimates. We also extend STTE--MDR to a doubly robust
estimator \emph{DR}--MDR \citep{bang2005doublyrobust,kallus2020drl}.

The simulations provide encouraging evidence for the proposed MDR
approach. An MDR or \emph{DR}-MDR estimator achieved the lowest RMSE in six of
the seven scenarios, although no single implementation was best in every
setting. The relative performance of the \emph{State} and \emph{Entry}
versions reflects a bias--variance tradeoff: \emph{State} estimators use
changing patient information over follow-up, whereas \emph{Entry}
estimators may be more stable with sparse data. In the nonlinear
scenario, which served as a model-misspecification test, MDR estimators
generally performed better than the non-MDR alternatives.

Several limitations should be noted when applying STTE--MDR.
Firstly, reducing the full history to a current state requires
state sufficiency, conditional-mean transportability, and state-level
overlap. A state that is too simple may omit relevant history, whereas a
state that is too detailed may produce weak overlap and unstable
weights. Like other methods that rely on conditional exchangeability,
STTE--MDR does not directly address unmeasured confounding. 
Double robustness
protects against misspecification of the outcome regression or the
density ratio, but not against misspecification of the target-state
generator.

The current method focuses on discrete time-to-event outcomes.
Future work could extend the risk-set transport framework to continuous
event times and non-absorbing binary outcomes, such as disease flares or
recurrent hospitalizations. The absence of an
empirical application also leaves the method's performance in complex
real-world data to be evaluated.

Statistical inference requires additional care because individuals may
enter several sequential trials and all nuisance components are
estimated. One possible approach is an individual-level bootstrap. In
each bootstrap sample, individuals would be resampled, their eligible
person-trials reconstructed, and all target-state, density-ratio, and
outcome models refitted. Resampling individuals rather than person-trial
rows preserves dependence across repeated trial entries. Bootstrap
confidence intervals could then be obtained from the empirical
distribution of the effect estimates, but their coverage was not
evaluated in this study and requires further investigation. Deriving an
asymptotically valid analytic variance estimator remains an open question for future research,
because uncertainty passes through several estimated components,
including target-state simulation and classifier-based weighting.

%% file: appendix.tex
\section{Proofs and Technical Details}
\label{app:theoretical-proofs}

This appendix proves the results stated in
Section~\ref{subsec:theoretical-properties}. Cellwise proofs use the
compact notation for a pooled strategy--follow-up cell \(c=(z,k)\),
with \(S_c\) denoting a generic augmented state as defined in the main
text. Formal longitudinal assumptions retain the explicit
\((i,m,k,z)\) indices from Section~\ref{subsec:setup-notation}.

\subsection{Formal Identification Assumptions}
\label{app:formal-identification-assumptions}

This subsection gives formal versions of the conditions summarized in
Section~\ref{subsubsec:identification-assumptions}. All notation follows
Section~\ref{subsec:setup-notation}. For \(k<\tau\), let
\[
  \mathcal U_{im,k+1:\tau}^{z}
  =
  \left(
    S_{im,k+1}^{z},\ldots,S_{im,\tau}^{z},
    \Delta Y_{im,k+1}^{z},\ldots,\Delta Y_{im,\tau}^{z}
  \right)
\]
collect the future counterfactual states and interval outcomes under
strategy \(z\). Equalities and conditional-independence statements
below are required only on histories with positive probability in the
corresponding risk set.

\paragraph{A1. No interference.}
For each individual \(i\), the counterfactual process
\(\mathcal U_{im,1:\tau}^{z}\) depends on that individual's treatment
strategy but not on the treatment assignments of other individuals.

\paragraph{A2. Consistency and well-defined strategies.}
Each strategy \(g_z\) assigns a unique treatment at every history to
which it applies. For a row satisfying \(J_{imk}^{z}=1\),
\begin{equation}
  J_{imk}^{z}=1
  \quad\Longrightarrow\quad
  \left(
    S_{imk},
    \Delta Y_{im,k+1}
  \right)
  =
  \left(
    S_{imk}^{z},
    \Delta Y_{im,k+1}^{z}
  \right).
  \label{eq:appendix-consistency}
\end{equation}

\paragraph{A3. Sequential treatment exchangeability.}
At each treatment decision,
\begin{equation}
  \mathcal U_{im,k+1:\tau}^{z}
  \perp\!\!\!\perp
  A_{imk}
  \mid
  \left(
    H_{imk},
    E_{im}=1,
    R_{imk}=1,
    O_{imk}=1,
    D_{im,k-1}^{z}=1
  \right).
  \label{eq:appendix-treatment-exchangeability}
\end{equation}
Thus, after conditioning on the measured predecision history, the
observed treatment decision contains no additional information about
the future counterfactual process under strategy \(z\).

\paragraph{A4. Sequential observation exchangeability.}
Because \(C_{im,k+1}\) is measured after the interval-\(k\) outcome,
the relevant condition concerns the subsequent state and future
outcomes among individuals who survive that interval. For
\(k=0,\ldots,\tau-2\),
\begin{equation}
  \left(
    S_{im,k+1}^{z},
    \mathcal U_{im,k+2:\tau}^{z}
  \right)
  \perp\!\!\!\perp
  C_{im,k+1}
  \mid
  \substack{
    H_{imk},\ A_{imk},\ E_{im}=1,\ R_{imk}=1,\\
    O_{imk}=1,\ D_{imk}^{z}=1,\ \Delta Y_{im,k+1}=0
  }.
  \label{eq:appendix-observation-exchangeability}
\end{equation}
Loss to follow-up after the final interval's outcome does not affect the
estimand at horizon \(\tau\).

\paragraph{A5. History-level positivity.}
For each target-relevant history and treatment decision,
\begin{equation}
  \Pr\!\left\{
    A_{imk}=g_z(H_{imk})
    \,\middle|\,
    \substack{
      H_{imk},\ E_{im}=1,\ R_{imk}=1,\ O_{imk}=1,\\
      D_{im,k-1}^{z}=1
    }
  \right\}
  >0.
  \label{eq:appendix-treatment-positivity}
\end{equation}
If elimination of loss to follow-up is part of the target intervention,
then, among survivors,
\begin{equation}
  \Pr\!\left(
    C_{im,k+1}=0
    \,\middle|\,
    \substack{
      H_{imk},\ E_{im}=1,\ R_{imk}=1,\ O_{imk}=1,\\
      D_{imk}^{z}=1,\ \Delta Y_{im,k+1}=0
    }
  \right)
  >0.
  \label{eq:appendix-observation-positivity}
\end{equation}

\paragraph{A6. State sufficiency and conditional-mean transportability.}
For \(\widetilde s=(m,s)\), define the target and observed-adherent
state-specific event risks by
\[
  \mu_{k,z}^{T}(\widetilde s)
  =
  E\!\left(
    \Delta Y_{im,k+1}^{z}
    \,\middle|\,
    S_{imk}^{z}=s,\,
    E_{im}=1,\,
    R_{imk}^{z}=1
  \right)
\]
and
\[
  \mu_{k,z}^{A}(\widetilde s)
  =
  E\!\left(
    \Delta Y_{im,k+1}
    \,\middle|\,
    S_{imk}=s,\,
    J_{imk}^{z}=1
  \right).
\]
State sufficiency requires
\begin{equation}
  E\!\left(
    \Delta Y_{im,k+1}^{z}
    \,\middle|\,
    H_{imk}^{z},
    E_{im}=1,
    R_{imk}^{z}=1
  \right)
  =
  \mu_{k,z}^{T}(\widetilde S_{imk}^{z}).
  \label{eq:appendix-counterfactual-state-sufficiency}
\end{equation}
Conditional-mean transportability requires
\begin{equation}
  \mu_{k,z}^{A}(\widetilde s)
  =
  \mu_{k,z}^{T}(\widetilde s)
  =:
  \mu_{k}^{z}(\widetilde s),
  \label{eq:appendix-conditional-mean-transportability}
\end{equation}
for \(P_{k,z}^{T,\omega}\)-almost every \(\widetilde s\). If the target law is propagated
from the compressed state rather than the full history, the
survivor-state transition law \(\mathcal{T}_{k,z}\) defined in
Equation~\eqref{eq:survivor-transition-law} must also satisfy, for every
measurable \(B\subseteq\widetilde{\mathcal S}\),
\begin{equation}
  \Pr\!\left(
    \widetilde S_{im,k+1}^{z}\in B
    \,\middle|\,
    \substack{
      H_{imk}^{z},\ E_{im}=1,\ R_{imk}^{z}=1,\\
      \Delta Y_{im,k+1}^{z}=0
    }
  \right)
  =
  \mathcal{T}_{k,z}\!\left(B\mid \widetilde S_{imk}^{z}\right).
  \label{eq:appendix-markov-sufficiency}
\end{equation}

\paragraph{A7. State-level overlap.}
For every target-relevant pooled strategy--follow-up cell,
\begin{equation}
  P_{k,z}^{T,\omega}\ll P_{k,z}^{A,\omega}.
  \label{eq:appendix-state-overlap}
\end{equation}
Because trial index \(m\) is part of the augmented state, this condition
requires observed-adherent support within the same trial for every
target-relevant augmented state.
The derivative
\(r_{k,z}^{\omega}=dP_{k,z}^{T,\omega}/dP_{k,z}^{A,\omega}\)
therefore exists \(P_{k,z}^{A,\omega}\)-almost surely. Population identification requires
integrability of the relevant weighted functions. Finite-variance
inference may additionally require
\begin{equation}
  E_{P_{k,z}^{A,\omega}}
  \left[
    \{r_{k,z}^{\omega}(\widetilde S)\}^{2}
  \right]
  <\infty.
  \label{eq:appendix-ratio-second-moment}
\end{equation}

\paragraph{A8. Correct temporal and risk-set alignment.}
The target and observed laws use the within-interval ordering
\[
  \left(
    L_{imk},
    A_{imk},
    \Delta Y_{im,k+1},
    C_{im,k+1}
  \right)
\]
and condition on being event-free at the beginning of the same
interval. For the absorbing event,
\begin{equation}
  R_{im,k+1}^{z}
  =
  R_{imk}^{z}
  \left(
    1-\Delta Y_{im,k+1}^{z}
  \right).
  \label{eq:appendix-risk-set-recursion}
\end{equation}
Consequently, a simulated individual who experiences the event during
interval \(k\) does not contribute a state to the target risk set at
interval \(k+1\).

\subsection{Derivation of the target risk-set recursion}
\label{app:target-risk-set-recursion}

Let \(\Pr_{\omega}^{z}\) denote the counterfactual law obtained by
mixing the eligible populations of the contributing trials with
baseline weights \(\omega_m\). Write
\(\widetilde S_k^{z}\), \(R_k^{z}\), and
\(\Delta Y_{k+1}^{z}\) for generic variables under this pooled law.
For any measurable \(B\subseteq\widetilde{\mathcal S}\), A8 implies
\begin{align}
  P_{k+1,z}^{T,\omega}(B)
  &=
  \Pr_{\omega}^{z}\!\left(
    \widetilde S_{k+1}^{z}\in B
    \,\middle|\,
    R_{k+1}^{z}=1
  \right) \\
  &=
  \frac{
    \Pr_{\omega}^{z}\!\left(
      \widetilde S_{k+1}^{z}\in B,
      \Delta Y_{k+1}^{z}=0
      \,\middle|\,
      R_{k}^{z}=1
    \right)
  }{
    \Pr_{\omega}^{z}\!\left(
      \Delta Y_{k+1}^{z}=0
      \,\middle|\,
      R_{k}^{z}=1
    \right)
  }.
  \label{eq:appendix-target-law-conditioning}
\end{align}
The denominator of
Equation~\eqref{eq:appendix-target-law-conditioning} is
\(1-\lambda_{k}^{z,\omega}\). By iterated expectation, its numerator
equals
\begin{align}
  &\int_{\widetilde{\mathcal S}}
  \Pr_{\omega}^{z}\!\left(
    \substack{
      \widetilde S_{k+1}^{z}\in B,\\
      \Delta Y_{k+1}^{z}=0
    }
    \,\middle|\,
    \widetilde S_{k}^{z}=\widetilde s,R_{k}^{z}=1
  \right)
  dP_{k,z}^{T,\omega}(\widetilde s) \\
  &\quad=
  \int_{\widetilde{\mathcal S}}
  \left\{1-\mu_{k,z}^{T}(\widetilde s)\right\}
  \mathcal T_{k,z}(B\mid\widetilde s)\,
  dP_{k,z}^{T,\omega}(\widetilde s).
  \label{eq:appendix-target-law-numerator}
\end{align}
Under A6, \(\mu_{k,z}^{T}=\mu_{k}^{z}\), yielding
Equation~\eqref{eq:target-law-recursion} whenever
\(1-\lambda_{k}^{z,\omega}>0\).

\subsection{Change of measure and balance}
\label{app:change-measure}

\begin{proof}[Proof of Proposition~\ref{prop:mdr-change-measure}]
Because \(P_c^T\ll P_c^A\), the Radon--Nikodym derivative
\(r_c=dP_c^T/dP_c^A\) exists and is unique
\(P_c^A\)-almost surely. Its defining integration rule gives
\[
  \int_{\widetilde{\mathcal S}}\phi(s)\,dP_c^T(s)
  =
  \int_{\widetilde{\mathcal S}}
  \phi(s)r_c(s)\,dP_c^A(s)
\]
for every \(P_c^T\)-integrable \(\phi\). Setting \(\phi(s)=1\) gives
\(E_{P_c^A}\{r_c(S_c)\}=1\). Setting
\(\phi(s)=I(s\in B)\) gives
\[
  E_{P_c^A}
  \{r_c(S_c)I(S_c\in B)\}
  =
  P_c^T(B)
\]
for every measurable \(B\subseteq\widetilde{\mathcal S}\).
\end{proof}

\subsection{Identification by direct MDR}
\label{app:mdr-identification}

\begin{proof}[Proof of Theorem~\ref{thm:mdr-identification}]
For \(c=(z,k)\), write
\(\mu_c^T=\mu_{k,z}^T\) and
\(\mu_c^A=\mu_{k,z}^A=\mu_c\). Then
\begin{equation}
  \lambda_c
  =
  \int_{\widetilde{\mathcal S}}
  \mu_c^T(s)\,dP_c^T(s).
  \label{eq:appendix-target-hazard-event-risk}
\end{equation}
By A7 and the Radon--Nikodym integration rule,
\begin{equation}
  \begin{aligned}
  E_{P_c^A}\{r_c(S_c)Y_c\}
  &=
  E_{P_c^A}
  \left[
    E\{r_c(S_c)Y_c\mid S_c\}
  \right]\\
  &=
  E_{P_c^A}
  \left[
    r_c(S_c)E(Y_c\mid S_c)
  \right]\\
  &=
  \int_{\widetilde{\mathcal S}}
    r_c(s)\mu_c^A(s)\,dP_c^A(s)\\
  &=
  \int_{\widetilde{\mathcal S}}
    \mu_c^A(s)\,dP_c^T(s).
  \end{aligned}
  \label{eq:appendix-mdr-observed-event-risk}
\end{equation}
Equations~\eqref{eq:appendix-target-hazard-event-risk}
and~\eqref{eq:appendix-mdr-observed-event-risk} imply
\begin{equation}
  E_{P_c^A}\{r_c(S_c)Y_c\}-\lambda_c
  =
  \int_{\widetilde{\mathcal S}}
  \{\mu_c^A(s)-\mu_c^T(s)\}\,dP_c^T(s).
  \label{eq:appendix-state-insufficiency-bias}
\end{equation}
By A6,
\(\mu_c^A(s)=\mu_c^T(s)\) \(P_c^T\)-almost surely; hence
\begin{equation}
  \lambda_c
  =E_{P_c^A}\{r_c(S_c)Y_c\}.
  \label{eq:appendix-direct-mdr-identification}
\end{equation}
Consequently,
\begin{equation}
  F_{\omega}^{z}(\tau)
  =1-\prod_{k=0}^{\tau-1}
    (1-\lambda_k^{z,\omega}),
  \qquad
  \psi_{\omega,\mathrm{RD}}(\tau)
  =F_{\omega}^{1}(\tau)-F_{\omega}^{0}(\tau).
\end{equation}
\end{proof}

\subsection{Classifier recovery of the density ratio}
\label{app:classifier-ratio}

\begin{proof}[Proof of Lemma~\ref{lem:classifier-density-ratio}]
Let \(p_c^T(s)\) and \(p_c^A(s)\) denote the densities of \(P_c^T\) and
\(P_c^A\), respectively, with respect to a common reference measure.
All probabilities in this proof are taken under the constructed
classifier training distribution. Let
\[
  \rho_c=\Pr(G=1\mid c).
\]
By Bayes' rule,
\begin{equation}
  \begin{aligned}
  h_c(s)
  &=\Pr(G=1\mid S_c=s,c)\\
  &=\frac{
      \rho_c p_c^T(s)
    }{
      \rho_c p_c^T(s)+(1-\rho_c)p_c^A(s)
    }.
  \end{aligned}
  \label{eq:appendix-classifier-posterior-odds}
\end{equation}
Hence, \(P_c^A\)-almost surely,
\begin{equation}
  r_c(s)
  =\frac{dP_c^T}{dP_c^A}(s)
  =\frac{p_c^T(s)}{p_c^A(s)}
  =\frac{1-\rho_c}{\rho_c}
    \frac{h_c(s)}{1-h_c(s)}.
  \label{eq:appendix-classifier-density-ratio}
\end{equation}
\end{proof}

For \(c=(z,k)\), \(\widetilde S=(m,S)\), and
\[
  n_c^T:=\#\{i:G_i=1,c_i=c\},
  \qquad
  n_c^A:=\#\{i:G_i=0,c_i=c\},
\]
\begin{equation}
  n_c^T=n_c^A
  \quad\Longrightarrow\quad
  \rho_c=\frac12
  \quad\Longrightarrow\quad
  r_c(s)=\frac{h_c(s)}{1-h_c(s)}.
  \label{eq:appendix-balanced-classifier-ratio}
\end{equation}

\subsection{Marginalization of the full-history ratio}
\label{app:rao-blackwell}

\begin{proof}[Proof of Theorem~\ref{thm:mdr-rao-blackwell}]
For measurable \(B\subseteq\widetilde{\mathcal S}\), let
\[
  s^{-1}(B)=\{h:s(h)\in B\}.
\]
Under \(S_c=s(H_c)\),
\[
  \begin{aligned}
  P_c^A(B)
  &=\Pr_{Q_c^A}(S_c\in B)\\
  &=\Pr_{Q_c^A}\{s(H_c)\in B\}\\
  &=\Pr_{Q_c^A}\{H_c\in s^{-1}(B)\}\\
  &=Q_c^A\{s^{-1}(B)\},\\
  P_c^T(B)
  &=Q_c^T\{s^{-1}(B)\}.
  \end{aligned}
\]
Assume \(Q_c^T\ll Q_c^A\). Then
\[
  \begin{aligned}
  P_c^A(B)=0
  &\quad\Longrightarrow\quad Q_c^A\{s^{-1}(B)\}=0\\
  &\quad\Longrightarrow\quad Q_c^T\{s^{-1}(B)\}=0\\
  &\quad\Longrightarrow\quad P_c^T(B)=0.
  \end{aligned}
\]
Hence,
\[
  Q_c^T\ll Q_c^A
  \quad\Longrightarrow\quad
  P_c^T\ll P_c^A.
\]
Therefore, \(r_c=dP_c^T/dP_c^A\) exists by Radon--Nikodym and is defined
\(P_c^A\)-almost surely.

For any bounded measurable \(f:\widetilde{\mathcal S}\to\mathbb R\),
\[
  \begin{aligned}
  E_{Q_c^A}\{W_c(H_c)f(S_c)\}
  &=\int W_c(h)f\{s(h)\}\,dQ_c^A(h)\\
  &=\int \frac{dQ_c^T}{dQ_c^A}(h)f\{s(h)\}\,dQ_c^A(h)\\
  &=\int f\{s(h)\}\,dQ_c^T(h)\\
  &=\int_{\widetilde{\mathcal S}} f(u)\,dP_c^T(u)\\
  &=\int_{\widetilde{\mathcal S}}
    \frac{dP_c^T}{dP_c^A}(u)f(u)\,dP_c^A(u)\\
  &=\int r_c\{s(h)\}f\{s(h)\}\,dQ_c^A(h)\\
  &=E_{Q_c^A}\{r_c(S_c)f(S_c)\}.
  \end{aligned}
\]
Taking \(f(u)=I(u\in B)\), for every measurable
\(B\subseteq\widetilde{\mathcal S}\),
\[
  E_{Q_c^A}
  \left[W_c(H_c)I(S_c\in B)\right]
  =
  E_{Q_c^A}
  \left[r_c(S_c)I(S_c\in B)\right].
\]
Let \(Q_c^A(dh\mid S_c=u)\) be a regular conditional law.
For every measurable \(B\subseteq\widetilde{\mathcal S}\),
\[
  \begin{aligned}
  \int_B
    \left\{\int W_c(h)\,Q_c^A(dh\mid S_c=u)\right\}
    dP_c^A(u)
  &=\int W_c(h)I\{s(h)\in B\}\,dQ_c^A(h)\\
  &=\int_{s^{-1}(B)}
    \frac{dQ_c^T}{dQ_c^A}(h)\,dQ_c^A(h)\\
  &=Q_c^T\{s^{-1}(B)\}\\
  &=P_c^T(B)\\
  &=\int_B \frac{dP_c^T}{dP_c^A}(u)\,dP_c^A(u)\\
  &=\int_B r_c(u)\,dP_c^A(u).
  \end{aligned}
\]
Therefore,
\[
  r_c(S_c)
  =E_{Q_c^A}\{W_c(H_c)\mid S_c\},
  \qquad Q_c^A\text{-almost surely}.
\]
If \(E_{Q_c^A}\{W_c(H_c)^2\}<\infty\), then
\[
  \begin{aligned}
  \operatorname{Var}_{Q_c^A}\{W_c(H_c)\}
  &=\operatorname{Var}_{Q_c^A}
    \left[E_{Q_c^A}\{W_c(H_c)\mid S_c\}\right]\\
  &\quad+E_{Q_c^A}
    \left[\operatorname{Var}_{Q_c^A}\{W_c(H_c)\mid S_c\}\right]\\
  &=\operatorname{Var}_{P_c^A}\{r_c(S_c)\}\\
  &\quad+E_{Q_c^A}
    \left[\operatorname{Var}_{Q_c^A}\{W_c(H_c)\mid S_c\}\right]\\
  &\geq \operatorname{Var}_{P_c^A}\{r_c(S_c)\}.
  \end{aligned}
\]
\end{proof}

\paragraph{Connection with cumulative IPACW.}
For \(c=(z,k)\), let \(h_j\) denote the history through visit
\(j\). Let \(\pi_{c,j}^{D}(h_j)\) and \(\pi_{c,j+1}^{O}(h_j)\) denote,
respectively, the conditional probabilities of adherence at visit \(j\)
and continued observation after visit \(j\), with the history and risk-set
conditioning specified in A5. 
Define
\[
  G_c(h)
  =
  \prod_{j=0}^{k}\pi_{c,j}^{D}(h_j)
  \prod_{j=0}^{k-1}\pi_{c,j+1}^{O}(h_j).
\]
Then
\[
  w_c^{\mathrm{IPACW}}(h)
  =G_c(h)^{-1}.
\]
Let \(J_c=1\) indicate inclusion in the observed-adherent risk set. Under
the identification assumptions,
\[
  \begin{aligned}
  G_c(h)
  &=\Pr(J_c=1\mid H_c=h),\\
  \Pr(H_c\in\mathcal B,J_c=1)
  &=\int_{\mathcal B}G_c(h)\,dQ_c^T(h),\\
  \Pr(J_c=1)
  &=\int G_c(h)\,dQ_c^T(h)
  =E_{Q_c^T}\{G_c(H_c)\}>0.
  \end{aligned}
\]
Consequently, for every measurable history set \(\mathcal B\),
\[
  \begin{aligned}
  Q_c^A(\mathcal B)
  &=\Pr(H_c\in\mathcal B\mid J_c=1)\\
  &=\frac{
    \Pr(H_c\in\mathcal B,J_c=1)
  }{
    \Pr(J_c=1)
  }\\
  &=\frac{
    \displaystyle\int_{\mathcal B}G_c(h)\,dQ_c^T(h)
  }{
    E_{Q_c^T}\{G_c(H_c)\}
  }.
  \end{aligned}
\]
By the definition of the Radon--Nikodym derivative, this means
\[
  \begin{aligned}
  \frac{dQ_c^A}{dQ_c^T}(h)
  &=\frac{G_c(h)}{E_{Q_c^T}\{G_c(H_c)\}},\\
  W_c(h)
  =\frac{dQ_c^T}{dQ_c^A}(h)
  &=\frac{E_{Q_c^T}\{G_c(H_c)\}}{G_c(h)}.
  \end{aligned}
\]
Moreover,
\[
  \begin{aligned}
  E_{Q_c^A}\{w_c^{\mathrm{IPACW}}(H_c)\}
  &=\int \frac{1}{G_c(h)}\,dQ_c^A(h)\\
  &=\frac{1}{E_{Q_c^T}\{G_c(H_c)\}}
    \int dQ_c^T(h)\\
  &=\frac{1}{E_{Q_c^T}\{G_c(H_c)\}},
  \end{aligned}
\]
and therefore
\[
  \frac{
    w_c^{\mathrm{IPACW}}(h)
  }{
    E_{Q_c^A}\{w_c^{\mathrm{IPACW}}(H_c)\}
  }
  =
  \frac{E_{Q_c^T}\{G_c(H_c)\}}{G_c(h)}
  =W_c(h).
\]
Hence,
\[
  \begin{aligned}
  r_c(S_c)
  &=E_{Q_c^A}\{W_c(H_c)\mid S_c\}\\
  &=\frac{
    E_{Q_c^A}\{w_c^{\mathrm{IPACW}}(H_c)\mid S_c\}
  }{
    E_{Q_c^A}\{w_c^{\mathrm{IPACW}}(H_c)\}
  }.
  \end{aligned}
\]
If \(E_{Q_c^A}\{w_c^{\mathrm{IPACW}}(H_c)^2\}<\infty\), then
\[
  \operatorname{Var}_{P_c^A}\{r_c(S_c)\}
  \leq
  \frac{
    \operatorname{Var}_{Q_c^A}\{w_c^{\mathrm{IPACW}}(H_c)\}
  }{
    \left[E_{Q_c^A}\{w_c^{\mathrm{IPACW}}(H_c)\}\right]^2
  }.
\]

\subsection{Conditional double robustness}
\label{app:double-robustness}

\begin{proof}[Proof of Theorem~\ref{thm:dr-mdr}]
Let \(\mathbb P_{A,c}\) denote the empirical law of the
observed-adherent cell and \(\mathbb P_{T,c}\) the empirical law of its
simulated target states. The unprojected estimator corresponding to
Equation~\eqref{eq:dr-mdr-hazard} is
\begin{equation}
  \widehat\Psi_c
  =
  \mathbb P_{T,c}\widehat\mu_c
  +
  \frac{
    \mathbb P_{A,c}
    \left[\widehat r_c\{Y_c-\widehat\mu_c\}\right]
  }{
    \mathbb P_{A,c}\widehat r_c
  },
  \qquad
  \widehat\lambda_c^{\mathrm{DR}}
  =\Pi_{[0,1]}(\widehat\Psi_c).
  \label{eq:appendix-dr-mdr-unprojected}
\end{equation}
Set
\[
  d_c=E_{P_c^A}\{\overline r_c(S_c)\},
  \qquad
  \overline r_c^{\,\dagger}=\frac{\overline r_c}{d_c}.
\]
The convergence assumptions in Theorem~\ref{thm:dr-mdr} give
\[
  \begin{aligned}
  \mathbb P_{T,c}\widehat\mu_c
  &\xrightarrow{p}
  E_{P_c^T}\{\overline\mu_c(S_c)\},\\
  \mathbb P_{A,c}\widehat r_c
  &\xrightarrow{p}
  E_{P_c^A}\{\overline r_c(S_c)\}=d_c,\\
  \mathbb P_{A,c}
  \left[\widehat r_c\{Y_c-\widehat\mu_c\}\right]
  &\xrightarrow{p}
  E_{P_c^A}
  \left[\overline r_c(S_c)
    \{Y_c-\overline\mu_c(S_c)\}\right].
  \end{aligned}
\]
Since \(0<d_c<\infty\), Slutsky's theorem gives
\[
  \frac{
    \mathbb P_{A,c}
    \left[\widehat r_c\{Y_c-\widehat\mu_c\}\right]
  }{
    \mathbb P_{A,c}\widehat r_c
  }
  \xrightarrow{p}
  E_{P_c^A}
  \left[
    \overline r_c^{\,\dagger}(S_c)
    \{Y_c-\overline\mu_c(S_c)\}
  \right].
\]
Therefore,
\begin{equation}
  \widehat\Psi_c
  \xrightarrow{p}
  \Psi_c(\overline\mu_c,\overline r_c),
  \label{eq:dr-mdr-probability-limit}
\end{equation}
where
\begin{equation}
  \Psi_c(\overline\mu_c,\overline r_c)
  =
  E_{P_c^T}\{\overline\mu_c(S_c)\}
  +
  E_{P_c^A}
  \left[
    \overline r_c^{\,\dagger}(S_c)
    \{Y_c-\overline\mu_c(S_c)\}
  \right].
  \label{eq:dr-mdr-population-functional}
\end{equation}
Since \(\overline r_c^{\,\dagger}(S_c)\) and
\(\overline\mu_c(S_c)\) are functions of \(S_c\), and
\(\mu_c(S_c)=E_{P_c^A}(Y_c\mid S_c)\), iterated expectation gives
\[
  \begin{aligned}
  &E_{P_c^A}
  \left[
    \overline r_c^{\,\dagger}(S_c)
    \{Y_c-\overline\mu_c(S_c)\}
  \right]\\
  &\quad=
  E_{P_c^A}
  \left[
    E_{P_c^A}
    \left\{
      \overline r_c^{\,\dagger}(S_c)
      \{Y_c-\overline\mu_c(S_c)\}
      \mid S_c
    \right\}
  \right]\\
  &\quad=
  E_{P_c^A}
  \left[
    \overline r_c^{\,\dagger}(S_c)
    \left\{
      E_{P_c^A}(Y_c\mid S_c)-\overline\mu_c(S_c)
    \right\}
  \right]\\
  &\quad=
  E_{P_c^A}
  \left[
    \overline r_c^{\,\dagger}(S_c)
    \{\mu_c(S_c)-\overline\mu_c(S_c)\}
  \right].
  \end{aligned}
\]
Proposition~\ref{prop:mdr-change-measure} also gives
\[
  E_{P_c^T}\{\overline\mu_c(S_c)\}
  =
  E_{P_c^A}\{r_c(S_c)\overline\mu_c(S_c)\},
  \qquad
  \lambda_c
  =
  E_{P_c^A}\{r_c(S_c)\mu_c(S_c)\}.
\]
Subtracting \(\lambda_c\) and collecting terms produces
\begin{equation}
  \Psi_c(\overline\mu_c,\overline r_c)-\lambda_c
  =
  E_{P_c^A}
  \left[
    \{\overline r_c^{\,\dagger}(S_c)-r_c(S_c)\}
    \{\mu_c(S_c)-\overline\mu_c(S_c)\}
  \right].
  \label{eq:dr-mdr-bias-remainder}
\end{equation}
The remainder is zero when
\(\overline\mu_c=\mu_c\). It is also zero when
\(\overline r_c=a_c r_c\) for \(a_c>0\), because
\(E_{P_c^A}(r_c)=1\) implies
\(\overline r_c^{\,\dagger}=r_c\).
Thus, under either condition,
\(\widehat\Psi_c\xrightarrow{p}\lambda_c\). By the continuous mapping
theorem, since \(\Pi_{[0,1]}\) is continuous and
\(\lambda_c\in[0,1]\),
\[
  \widehat\lambda_c^{\mathrm{DR}}
  =\Pi_{[0,1]}(\widehat\Psi_c)
  \xrightarrow{p}
  \Pi_{[0,1]}(\lambda_c)
  =\lambda_c.
\]
\end{proof}

\subsubsection{Target-law qualification}
\label{app:target-law-error}

Let \(\overline P_c^T\) be the probability limit of a possibly
misspecified target-state generator. Using the preceding
iterated-expectation identity,
\(\lambda_c=E_{P_c^T}\{\mu_c(S_c)\}\), and
Proposition~\ref{prop:mdr-change-measure}, the error decomposes as
\begin{equation}
  \begin{aligned}
  &\Psi_c(
    \overline\mu_c,\overline r_c;
    \overline P_c^T
  )
  -\lambda_c\\
  &=
  E_{\overline P_c^T}\{\overline\mu_c(S_c)\}
  +E_{P_c^A}
  \left[
    \overline r_c^{\,\dagger}(S_c)
    \{\mu_c(S_c)-\overline\mu_c(S_c)\}
  \right]
  -E_{P_c^T}\{\mu_c(S_c)\}\\
  &=
  \left[
    E_{\overline P_c^T}\{\overline\mu_c(S_c)\}
    -E_{P_c^T}\{\overline\mu_c(S_c)\}
  \right]\\
  &\quad+
  E_{P_c^T}
  \{\overline\mu_c(S_c)-\mu_c(S_c)\}
  +E_{P_c^A}
  \left[
    \overline r_c^{\,\dagger}(S_c)
    \{\mu_c(S_c)-\overline\mu_c(S_c)\}
  \right]\\
  &=
  \left[
    E_{\overline P_c^T}\{\overline\mu_c(S_c)\}
    -E_{P_c^T}\{\overline\mu_c(S_c)\}
  \right]\\
  &\quad+
  E_{P_c^A}
  \left[
    r_c(S_c)\{\overline\mu_c(S_c)-\mu_c(S_c)\}
    +\overline r_c^{\,\dagger}(S_c)
    \{\mu_c(S_c)-\overline\mu_c(S_c)\}
  \right]\\
  &=
  \left[
    E_{\overline P_c^T}\{\overline\mu_c(S_c)\}
    -E_{P_c^T}\{\overline\mu_c(S_c)\}
  \right]\\
  &\quad+
  E_{P_c^A}
  \left[
    \{\overline r_c^{\,\dagger}(S_c)-r_c(S_c)\}
    \{\mu_c(S_c)-\overline\mu_c(S_c)\}
  \right].
  \end{aligned}
  \label{eq:appendix-target-law-error}
\end{equation}
The difference in square brackets is target-generator error and is not
removed by correctly specifying either the outcome regression or the
density ratio.
This establishes why Theorem~\ref{thm:dr-mdr} is conditional on
consistent target-risk-set generation and why the estimator is not
triply robust.

\subsubsection{Nuisance-estimation remainder and target-generation error}
\label{app:dr-remainder}

Under the correct target-state law, the DR-MDR bias remainder is
\begin{equation}
  E_{P_c^A}
  \left[
    \{\overline r_c^{\,\dagger}(S_c)-r_c(S_c)\}
    \{\mu_c(S_c)-\overline\mu_c(S_c)\}
  \right].
  \label{eq:appendix-dr-remainder}
\end{equation}
By the Cauchy--Schwarz inequality,
\begin{equation}
  \begin{aligned}
  &\left|
    E_{P_c^A}
    \left[
      \{\overline r_c^{\,\dagger}(S_c)-r_c(S_c)\}
      \{\mu_c(S_c)-\overline\mu_c(S_c)\}
    \right]
  \right|\\
  &\quad\leq
    \left[
      E_{P_c^A}
      \left\{
        \overline r_c^{\,\dagger}(S_c)-r_c(S_c)
      \right\}^{2}
    \right]^{1/2}\\
  &\qquad\quad\times
    \left[
      E_{P_c^A}
      \left\{
        \overline\mu_c(S_c)-\mu_c(S_c)
      \right\}^{2}
    \right]^{1/2}.
  \end{aligned}
  \label{eq:appendix-dr-product-bound}
\end{equation}
Following the product-rate condition in
\citet{chernozhukov2018dml}, the remainder is \(o_p(n^{-1/2})\) if
\begin{equation}
  \begin{aligned}
  &\left[
    E_{P_c^A}
    \left\{
      \widehat r_c^{\,\dagger}(S_c)-r_c(S_c)
    \right\}^{2}
  \right]^{1/2}\\
  &\quad\times
  \left[
    E_{P_c^A}
    \left\{
      \widehat\mu_c(S_c)-\mu_c(S_c)
    \right\}^{2}
  \right]^{1/2}
  =o_p(n^{-1/2}).
  \end{aligned}
  \label{eq:appendix-dr-product-rate}
\end{equation}
This condition controls the nuisance-product remainder; \(\sqrt{n}\)-rate
inference additionally requires appropriate regularity conditions and
first-order control of target-law estimation and target-simulation
error.

Let \(\widehat P_{c,\mathrm{gen}}^T\) denote the target-state law implied
by the fitted transition models
\citep{taubman2009parametricgformula,young2011parametricgformula,mcgrath2020gformula}.
Given \(N_c\) target states
\(S_{c1}^T,\ldots,S_{cN_c}^T\) simulated from this law, define
\begin{equation}
  \mathbb P_{T,c}\{\widehat\mu_c(S_c)\}
  =
  \frac{1}{N_c}\sum_{b=1}^{N_c}\widehat\mu_c(S_{cb}^T).
  \label{eq:appendix-target-monte-carlo-mean}
\end{equation}
The target-integration error decomposes as
\begin{equation}
  \begin{aligned}
  &\mathbb P_{T,c}\{\widehat\mu_c(S_c)\}
  -E_{P_c^T}\{\widehat\mu_c(S_c)\}\\
  &=
  \left[
    \mathbb P_{T,c}\{\widehat\mu_c(S_c)\}
    -E_{\widehat P_{c,\mathrm{gen}}^T}
      \{\widehat\mu_c(S_c)\}
  \right]\\
  &\quad+
  \left[
    E_{\widehat P_{c,\mathrm{gen}}^T}
      \{\widehat\mu_c(S_c)\}
    -E_{P_c^T}\{\widehat\mu_c(S_c)\}
  \right].
  \end{aligned}
  \label{eq:appendix-target-error-decomposition}
\end{equation}
By the Monte Carlo central limit
theorem:
\begin{equation}
  \mathbb P_{T,c}\{\widehat\mu_c(S_c)\}
  -
  E_{\widehat P_{c,\mathrm{gen}}^T}
  \{\widehat\mu_c(S_c)\}
  =
  O_p(N_c^{-1/2}).
  \label{eq:appendix-target-monte-carlo-rate}
\end{equation}

If
\(\widehat P_{c,\mathrm{gen}}^T\xrightarrow{p}\overline P_c^T\) and
\(\widehat\mu_c\xrightarrow{p}\overline\mu_c\), then
\begin{equation}
  E_{\widehat P_{c,\mathrm{gen}}^T}
    \{\widehat\mu_c(S_c)\}
  -E_{P_c^T}\{\widehat\mu_c(S_c)\}
  \xrightarrow{p}
  E_{\overline P_c^T}\{\overline\mu_c(S_c)\}
  -E_{P_c^T}\{\overline\mu_c(S_c)\},
  \label{eq:appendix-target-generator-limit}
\end{equation}
which equals zero when \(\overline P_c^T=P_c^T\)
\citep{chiu2023gformulaspecification}. Moreover,
\begin{equation}
  \begin{aligned}
  N_c\asymp n
  \quad\Longrightarrow\quad
  &\mathbb P_{T,c}\{\widehat\mu_c(S_c)\}
  -E_{\widehat P_{c,\mathrm{gen}}^T}
    \{\widehat\mu_c(S_c)\}\\
  &=O_p(n^{-1/2}).
  \end{aligned}
  \label{eq:appendix-target-monte-carlo-scaling}
\end{equation}

A subject-level bootstrap should repeat target-model fitting,
target-state simulation, classifier fitting, outcome-model fitting, and
the final augmentation in every replicate
\citep{westreich2012parametricgformula}. 

\subsubsection{Propagation to cumulative risks}
\label{app:dr-cumulative-risk}

Suppose \(\tau<\infty\) and, for every target-positive cell
\(c=(z,k)\),
\begin{equation}
  \widehat P_{c,\mathrm{gen}}^T\xrightarrow{p}P_c^T,
  \qquad
  \left\{
  \overline\mu_c=\mu_c
  \right\}
  \ \lor\ 
  \left\{
  \overline r_c=a_c r_c\text{ for some }a_c>0
  \right\}.
  \label{eq:appendix-cellwise-dr-conditions}
\end{equation}
By Theorem~\ref{thm:dr-mdr},
\begin{equation}
  \widehat\lambda_k^{z,\omega,\mathrm{DR}}
  \xrightarrow{p}
  \lambda_k^{z,\omega},
  \qquad
  z\in\{0,1\},\quad k=0,\ldots,\tau-1.
  \label{eq:appendix-cellwise-hazard-convergence}
\end{equation}
Let
\begin{equation}
  \boldsymbol\lambda^z
  =
  (\lambda_0^{z,\omega},\ldots,
   \lambda_{\tau-1}^{z,\omega})^\top,
  \qquad
  g(\boldsymbol\lambda^z)
  =
  1-\prod_{k=0}^{\tau-1}(1-\lambda_k^{z,\omega}).
  \label{eq:appendix-cumulative-risk-map}
\end{equation}
Since \(g\) is continuous,
\begin{equation}
  \widehat F_{\omega,\mathrm{DR}}^z(\tau)
  =g(\widehat{\boldsymbol\lambda}^{z,\mathrm{DR}})
  \xrightarrow{p}
  g(\boldsymbol\lambda^z)
  =F_\omega^z(\tau),
  \qquad z\in\{0,1\},
  \label{eq:appendix-cumulative-risk-convergence}
\end{equation}
and
\begin{equation}
  \widehat\psi_{\omega,\mathrm{DR}}(\tau)
  =
  \widehat F_{\omega,\mathrm{DR}}^1(\tau)
  -\widehat F_{\omega,\mathrm{DR}}^0(\tau)
  \xrightarrow{p}
  F_\omega^1(\tau)-F_\omega^0(\tau)
  =\psi_{\omega,\mathrm{RD}}(\tau).
  \label{eq:appendix-risk-difference-convergence}
\end{equation}

\clearpage
\section{Detailed Simulation Design}
\label{app:simulation-design}

\subsection{Reference longitudinal data-generating process}

For each simulated individual \(i=1,\ldots,n\), the time-invariant
covariates were generated as
\begin{align}
  X_{i1} &\sim \operatorname{Bernoulli}(0.5), \nonumber\\
  X_{i2} &\sim N(0,1), \nonumber\\
  X_{i3}\mid X_{i1},X_{i2}
  &\sim
  \operatorname{Bernoulli}\!\left[
    \operatorname{expit}
    \left\{-0.20+0.30X_{i1}+0.50X_{i2}\right\}
  \right].
  \label{eq:simulation-baseline-covariates}
\end{align}
The time-varying covariate at cohort baseline was
\begin{equation}
  L_{i0}
  =
  0.50X_{i1}
  +0.70X_{i2}
  +0.40X_{i3}
  +\varepsilon_{i0},
  \qquad
  \varepsilon_{i0}\sim N(0,1).
  \label{eq:simulation-baseline-L}
\end{equation}

At the beginning of month \(t=0,\ldots,23\), treatment was generated
among individuals who were event-free and uncensored according to
\begin{align}
  \operatorname{logit}
  \Pr(A_{it}=1\mid A_{i,t-1},L_{it},X_i)
  ={}&
  -1.20
  +2.00A_{i,t-1}
  +0.80L_{it} \nonumber\\
  &+0.30X_{i1}
  +0.20X_{i2}
  +0.10X_{i3}
  +0.02t,
  \label{eq:simulation-treatment-model}
\end{align}
where \(A_{i,-1}=0\). Conditional on being at risk at the start of the
interval, the event hazard was
\begin{align}
  \operatorname{logit}
  \Pr(\Delta Y_{i,t+1}=1\mid A_{it},L_{it},X_i)
  ={}&
  -6.50
  -0.40A_{it}
  +0.80L_{it} \nonumber\\
  &+0.30X_{i1}
  +0.20X_{i2}
  +0.20X_{i3}
  +0.03t.
  \label{eq:simulation-outcome-model}
\end{align}
Among individuals who remained event-free through the interval, loss to
follow-up was generated from
\begin{align}
  \operatorname{logit}
  \Pr(\Delta C_{i,t+1}=1\mid A_{it},L_{it},X_i)
  ={}&
  -5.00
  +0.20A_{it}
  +0.40L_{it} \nonumber\\
  &+0.20X_{i1}
  +0.20X_{i2}
  +0.10X_{i3}.
  \label{eq:simulation-censoring-model}
\end{align}
Finally, for individuals who remained event-free and uncensored, the
next covariate value was generated as
\begin{align}
  L_{i,t+1}
  ={}&
  0.60L_{it}
  +0.30X_{i1}
  +0.40X_{i2}
  +0.20X_{i3} \nonumber\\
  &-0.50A_{it}
  +\varepsilon_{i,t+1},
  \qquad
  \varepsilon_{i,t+1}\sim N(0,1).
  \label{eq:simulation-L-transition}
\end{align}
Thus, events were generated before loss to follow-up, and the next
covariate value was used only for individuals remaining under
observation. Unless otherwise specified for a simulation scenario,
treatment, event, and censoring probabilities were restricted to
\([0.01,0.99]\).

\subsection{Monte Carlo target value}

For each scenario, the true causal contrast was approximated using a
separate Monte Carlo population of \(50{,}000\) individuals. We first
simulated their natural histories without loss to follow-up and
identified all eligible person-trials. Starting from each eligible
trial-entry state, we then simulated 12-month outcomes under \(g_1\)
and \(g_0\). Paired random-number streams were used to reduce simulation
noise in the contrast. If \(\mathcal E\) denotes the resulting collection
of eligible person-trials, the target risks and risk difference were
\begin{equation}
  F^{z}(12)
  =
  \frac{1}{|\mathcal E|}
  \sum_{(i,m)\in\mathcal E}
  Y_{im,12}^{z},
  \qquad
  \psi(12)
  =
  F^{1}(12)-F^{0}(12).
  \label{eq:simulation-monte-carlo-truth}
\end{equation}
Each eligible person-trial therefore received equal weight, matching the
pooled sequential-trial estimand used in the analysis.

\clearpage
\section{Detailed Simulation Estimator Specifications}
\label{app:simulation-estimator-details}

This section supplements the estimator overview in
Section~\ref{subsec:estimators-compared}. All estimators targeted the
same pooled 12-month per-protocol risk difference.

\subsection{Target-state simulation}

The IPCW--State and DR-IPCW--State estimators and all four MDR estimators
used the strategy-specific target risk sets described in the main text.
The IPCW--Entry and DR-IPCW--Entry estimators did not require target-state
simulation. For each strategy, the simulated target sample was ten times
the number of eligible person-trials.

The covariate-transition model in
Equation~\eqref{eq:simulation-mdr-transition-model} was fitted by linear
regression among contiguous event-free and uncensored observations.
Forward simulation used Gaussian residuals with the fitted residual
standard deviation. The event-depletion model in
Equation~\eqref{eq:simulation-mdr-target-event-model} was fitted by
logistic regression among observations at risk. At the first simulated
step, \(L_{i,t-1}\) was initialized to the current trial-entry value.
Treatment was fixed to \(A_{it}=z\), individuals experiencing an event
were removed from subsequent risk sets, and \(L_{i,t+1}\) was simulated
only for those remaining event-free.

\subsection{MDR implementation}

All four MDR estimators used the follow-up-balanced, cell-normalized
density-ratio construction in
Section~\ref{subsec:mdr-weighted-algorithm}. A separate
\texttt{ranger} probability forest was fitted for each strategy, pooling
all follow-up months. Within each month, target rows were sampled to
match the number of observed-adherent rows. Each forest used 500 trees;
the remaining forest tuning parameters used the package defaults
\citep{breiman2001randomforests}.

Classifier probabilities from
Equation~\eqref{eq:simulation-mdr-classifier} were bounded to
\([0.001,0.999]\). The resulting odds were converted to the normalized
MDR weights in Equations~\eqref{eq:classifier-mdr}
and~\eqref{eq:cell-normalized-mdr}, with no post-estimation truncation.
No separate natural-censoring weight was multiplied by the MDR because
the density ratio directly transported the observed-adherent risk set to
the no-censoring target risk set. The MDR nuisance models were fitted on
the full simulated sample without cross-fitting.

For MDR--Entry and MDR--State, the outcome models in
Equations~\eqref{eq:simulation-entry-outcome-model}
and~\eqref{eq:simulation-state-outcome-model} were fitted using the MDR
weights. DR-MDR--Entry and DR-MDR--State fitted the same outcome models
without outcome weights and used the MDR weights in the residual
correction. For the State version, correction was performed on the
cell-specific hazard scale before the hazards were combined into
cumulative risks.

\subsection{IPCW implementation}

Because the strategies had no grace period, remaining compatible with
strategy \(z\) at month \(j\) was equivalent to receiving
\(A_{imj}=z\). The treatment models in
Equation~\eqref{eq:simulation-stabilized-artificial-weight} were fitted
among eligible person-trial intervals that were event-free and
uncensored at the start of month \(j\) and compatible with strategy
\(z\) through month \(j-1\). The natural-censoring models in
Equation~\eqref{eq:simulation-stabilized-natural-weight} were fitted on
the subset that also remained compatible at month \(j\) and event-free
during interval \(j\); equivalently, these rows satisfied
\(J_{imj}^{z}=1\) and \(\Delta Y_{im,j+1}=0\). All probabilities were
estimated by logistic regression, with \(j\) represented by indicators.

The artificial-censoring product included factors from \(j=0\) through
\(k\). The natural-censoring product ended at \(k-1\) because the event
was observed before natural censoring within month \(k\); this component
equaled one at \(k=0\). Estimated probabilities were bounded to
\([0.01,0.99]\), and the resulting weights were not truncated.

IPCW--Entry and IPCW--State fitted the appropriate outcome model in
Equations~\eqref{eq:simulation-entry-outcome-model}
and~\eqref{eq:simulation-state-outcome-model} using the stabilized
weights. DR-IPCW--Entry and DR-IPCW--State fitted the same formulas
without outcome weights and incorporated the stabilized weights through
the residual correction.

\subsection{TrialEmulation}

This per-protocol implementation censored individuals when they deviated
from the treatment assigned at trial entry. For a trial beginning at
month \(m\), the cohort month corresponding to follow-up month \(j\) is
\(m+j\). Numerator and denominator models were fitted on the original
cohort records supplied to the package and separately according to
treatment at the previous visit. The stabilized switching weight for
strategy \(z\) was
\begin{align}
  \widehat{\mathrm{SW}}_{imk}^{\mathrm{TE},A}(z)
  ={}&
  \prod_{j=1}^{k}
  \frac{
    \widehat{\Pr}\!\left(
      A_{imj}=z
      \,\middle|\,
      A_{im,j-1}=z,X_i
    \right)
  }{
    \widehat{\Pr}\!\left(
      A_{imj}=z
      \,\middle|\,
      A_{im,j-1}=z,X_i,L_{imj},m+j
    \right)
  }.
  \label{eq:simulation-trialemulation-switch-weight}
\end{align}
The corresponding natural-censoring weight was
\begin{align}
  \widehat{\mathrm{SW}}_{imk}^{\mathrm{TE},C}(z)
  ={}&
  \prod_{j=1}^{k}
  \frac{
    \widehat{\Pr}\!\left(
      \Delta C_{im,j+1}=0
      \,\middle|\,
      A_{im,j-1}=z,X_i
    \right)
  }{
    \widehat{\Pr}\!\left(
      \Delta C_{im,j+1}=0
      \,\middle|\,
      A_{im,j-1}=z,X_i,L_{imj},m+j
    \right)
  }.
  \label{eq:simulation-trialemulation-censor-weight}
\end{align}
The outcome analysis used their product,
\begin{equation}
  \widehat{\mathrm{SW}}_{imk}^{\mathrm{TE}}(z)
  =
  \widehat{\mathrm{SW}}_{imk}^{\mathrm{TE},A}(z)
  \widehat{\mathrm{SW}}_{imk}^{\mathrm{TE},C}(z).
  \label{eq:simulation-trialemulation-combined-weight}
\end{equation}
Both components equal one at trial entry. The package censoring models
were not additionally restricted to rows with
\(\Delta Y_{im,j+1}=0\), and the implementation applied no additional
probability bounds or post-estimation weight truncation. The outcome
model was the Entry specification in
Equation~\eqref{eq:simulation-entry-outcome-model}.

TrialEmulation and IPCW--Entry both used stabilized weights. They also
used the same weighted pooled-logistic outcome model, but their weight
alignment differed. IPCW--Entry included artificial-censoring factors from \(j=0\)
through \(k\) and natural-censoring factors from \(j=0\) through \(k-1\).
TrialEmulation used both factors from \(j=1\) through \(k\), thereby
treating entry treatment as the assigned strategy without reweighting it.
Its censoring product excluded the entry-interval factor but included
the current-interval factor. IPCW--Entry instead carried the
entry-interval censoring factor forward and excluded censoring after the
current outcome. IPCW--Entry fitted models using follow-up month \(j\)
and trial-entry month \(m\), whereas TrialEmulation used the original
cohort time scale \(m+j\).

\subsection{Longitudinal TMLE--GLM}

Longitudinal TMLE was applied to each eligible person-trial history in
wide format. Let \(W_{im}=(m,X_i,L_{im})\) denote the trial-entry
variables. In R-style notation, its node-specific working models were
\begin{align*}
  A_{im0}
  &\sim W_{im},\\
  A_{imj}
  &\sim W_{im}+L_{imj}+A_{im,j-1},
  \qquad j\geq1,\\
  \mathbf{1}(\Delta C_{im,j+1}=0)
  &\sim W_{im}+L_{imj}+A_{imj},\\
  Q_{im,j+1}
  &\sim W_{im}+L_{imj}+A_{imj}.
\end{align*}
Here \(Q_{im,j+1}\) denotes the iterated outcome regression at each
outcome or covariate node. At \(j=0\), \(L_{im0}=L_{im}\) was already
included in \(W_{im}\). All components used generalized linear models.

We used the package's targeted TMLE option. We did not use its untargeted
g-computation estimator. The initial sequential \(Q\)-regressions supplied
the g-computation component, and the targeting step incorporated the
estimated treatment and censoring mechanisms. Under the usual
identification and regularity conditions, the estimator is consistent
if either all sequential outcome regressions or all treatment and
censoring models are correctly specified \citep{lendle2017ltmle}.
Treatment was intervened on through the always-treated and never-treated
regimes, natural censoring was represented by censoring nodes, and the
cumulative treatment-and-censoring mechanism was bounded to
\([0.01,1]\).

\subsection{Methods outside the comparison}

The matched Entry and State specifications were used to isolate the
effect of weight construction when comparing MDR and IPCW. Longitudinal
CBPS, joint calibration, residual balancing, and kernel optimal weighting
were not included because adapting their published formulations to
overlapping sequential trials, artificial censoring, and
strategy-specific survivor risk sets would require additional
method-specific choices outside the scope of this study.

\clearpage
\section{Additional Simulation-Scenario Diagnostics}
\label{app:scenario-diagnostics}

This section provides the exact parameter changes, diagnostic
definitions, and supporting results for the scenarios summarized in
Section~\ref{subsec:simulation-scenarios}. Because the data-generating
process (DGP) was known, the DGP-implied probabilities and odds ratios
were calculated directly from its specified models. In the strong-confounding
scenario, the coefficient of \(L_{it}\) increased from \(0.80\) to
\(1.20\) in both the treatment and event models. In the poor-positivity
scenario, its coefficient in the treatment model increased from \(0.80\)
to \(2.00\), while its event-model coefficient remained \(0.80\). In the
poor-adherence scenario, the coefficient of \(A_{i,t-1}\) in the
treatment model decreased from \(2.00\) to \(0.60\). In the rare-event
scenario, the event-model intercept decreased from \(-6.50\) to
\(-8.00\). Finally, in the nonlinear scenario, the treatment linear
predictor included \(0.35L_{it}^{2}\), and the event linear predictor
included \(0.20L_{it}^{2}-0.25A_{it}L_{it}\).

Treatment, event, and censoring probabilities were restricted to
\([0.01,0.99]\) in all scenarios except the rare-event scenario, for
which the lower event-probability bound was \(10^{-4}\). Thus, the
poor-positivity scenario represents severe practical nonoverlap rather
than structural nonpositivity.

We assessed confounding by examining how strongly \(L_{it}\) was related
to both treatment assignment and event risk. Using the known DGP, we
compared the treatment and event odds when
\(L_{it}\) was set to a relatively high value (the third quartile) versus
a relatively low value (the first quartile), while keeping all other
variables unchanged.
To assess practical positivity, we examined how often treatment was
either very unlikely or almost certain. Each person contributed one
observation for every month they remained under follow-up. We reported
the percentage of these monthly observations for which the treatment
probability was at most \(0.05\) or at least \(0.95\). Adherence was
measured as the percentage of people who remained on treatment from one
month to the next. Event frequency was measured as the estimated
percentage of people who experienced the event within 12 months.

Table~\ref{tab:simulation-scenario-diagnostics} reports the numerical
summaries. Figure~\ref{fig:appendix-positivity-distributions} shows the
full distribution of the DGP-implied treatment-assignment probabilities, and
Figure~\ref{fig:appendix-event-risk-curves} shows how cumulative event
risk evolved over the complete 24-month observational period.

\input{scenario_diagnostics/table_scenario_diagnostics}

\begin{figure}[p]
  \centering
  \includegraphics[width=0.8\textwidth]{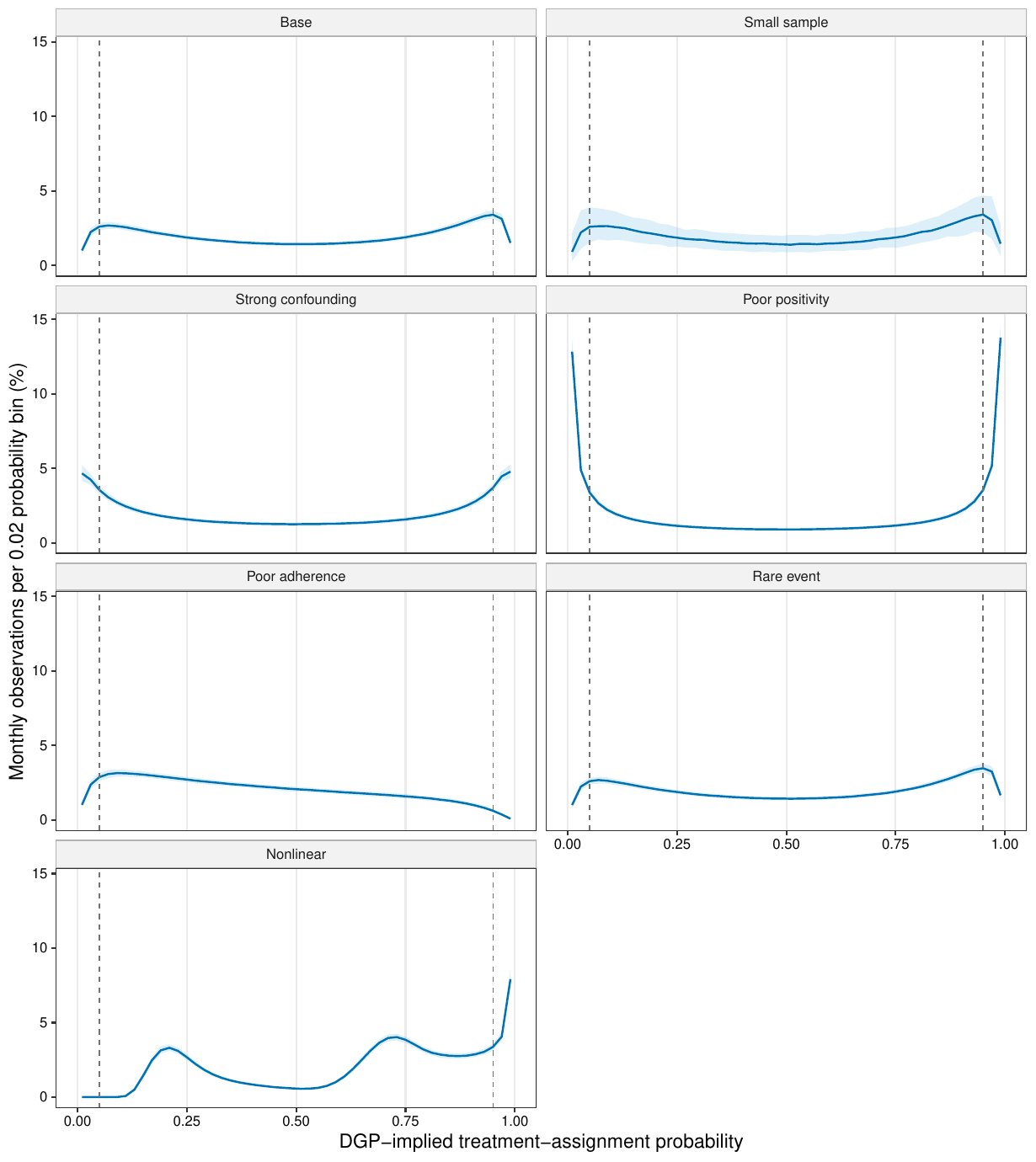}
  \caption{Distribution of the DGP-implied treatment-assignment
  probability across observed histories. Lines show the median percentage
  of monthly observations in each probability bin, and ribbons show empirical
  2.5th and 97.5th percentiles across 500 replications. Dashed vertical
  lines mark \(0.05\) and \(0.95\), the thresholds used to define
  practical positivity violations. DGP denotes the data-generating process;
  these probabilities were calculated directly from the known simulation
  treatment model.}
  \label{fig:appendix-positivity-distributions}
\end{figure}

\begin{figure}[p]
  \centering
  \includegraphics[width=0.8\textwidth]{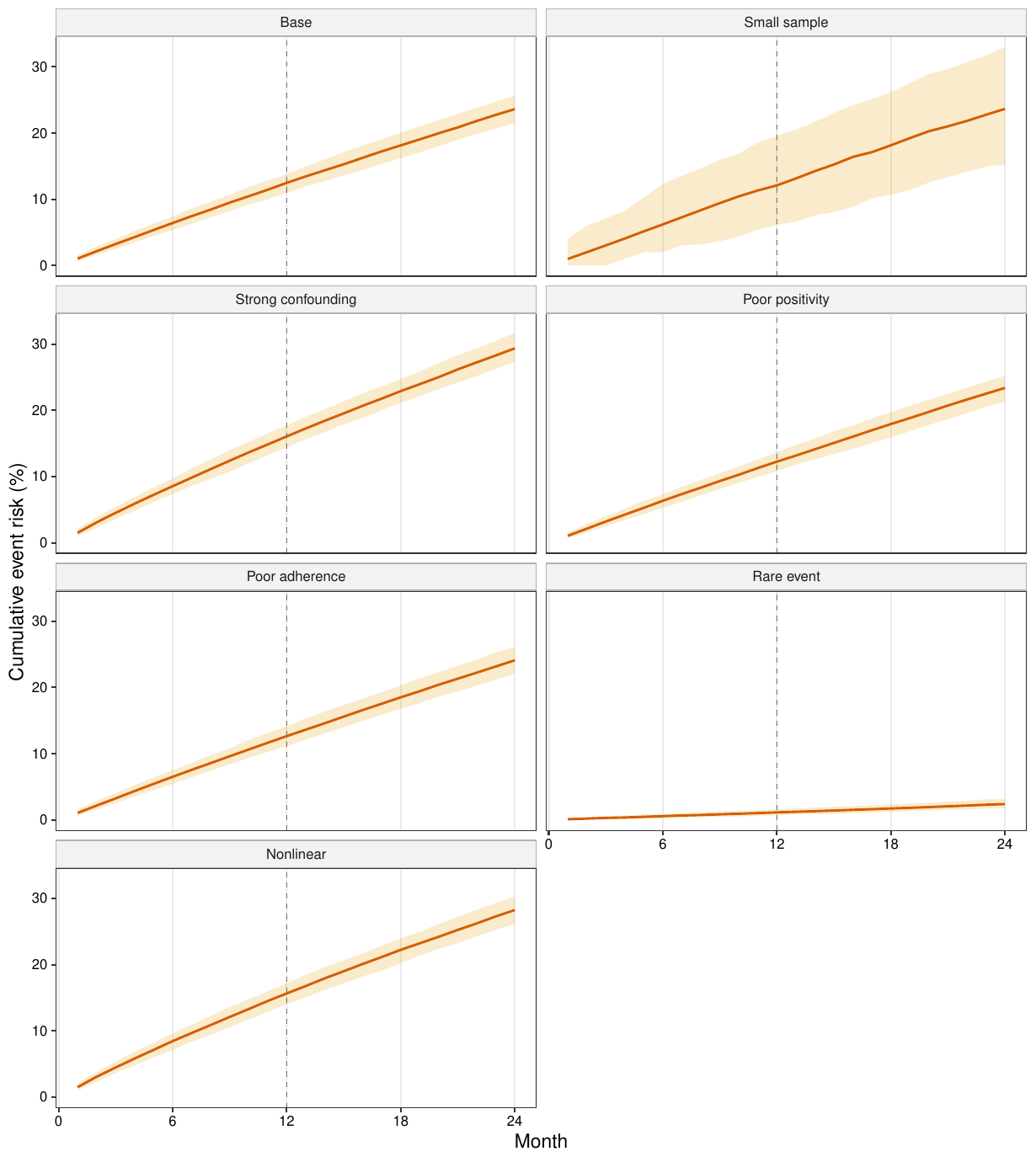}
  \caption{Observed cumulative event risk over the 24-month
  observational period. Lines show replication medians, and ribbons
  show empirical 2.5th and 97.5th percentiles across 500 replications.
  The dashed vertical line marks the 12-month target-trial follow-up
  horizon.}
  \label{fig:appendix-event-risk-curves}
\end{figure}

\clearpage
\section{Additional Simulation Weight Diagnostics}
\label{app:simulation-weight-diagnostics}

Figures~\ref{fig:appendix-weight-diagnostics},
\ref{fig:appendix-weight-max-distribution}, and
\ref{fig:appendix-weight-mean-distribution} summarize the stability of
the three explicit weight constructions. The IPCW values apply to
IPCW--Entry, IPCW--State, DR-IPCW--Entry, and DR-IPCW--State; the MDR
values likewise apply to the four MDR estimators. Longitudinal
TMLE--GLM is not shown because its cumulative treatment-and-censoring
mechanism diagnostics are not directly comparable with the row-level
weights used by the other estimators.

\begin{figure}[htbp]
  \centering
  \includegraphics[width=\textwidth]{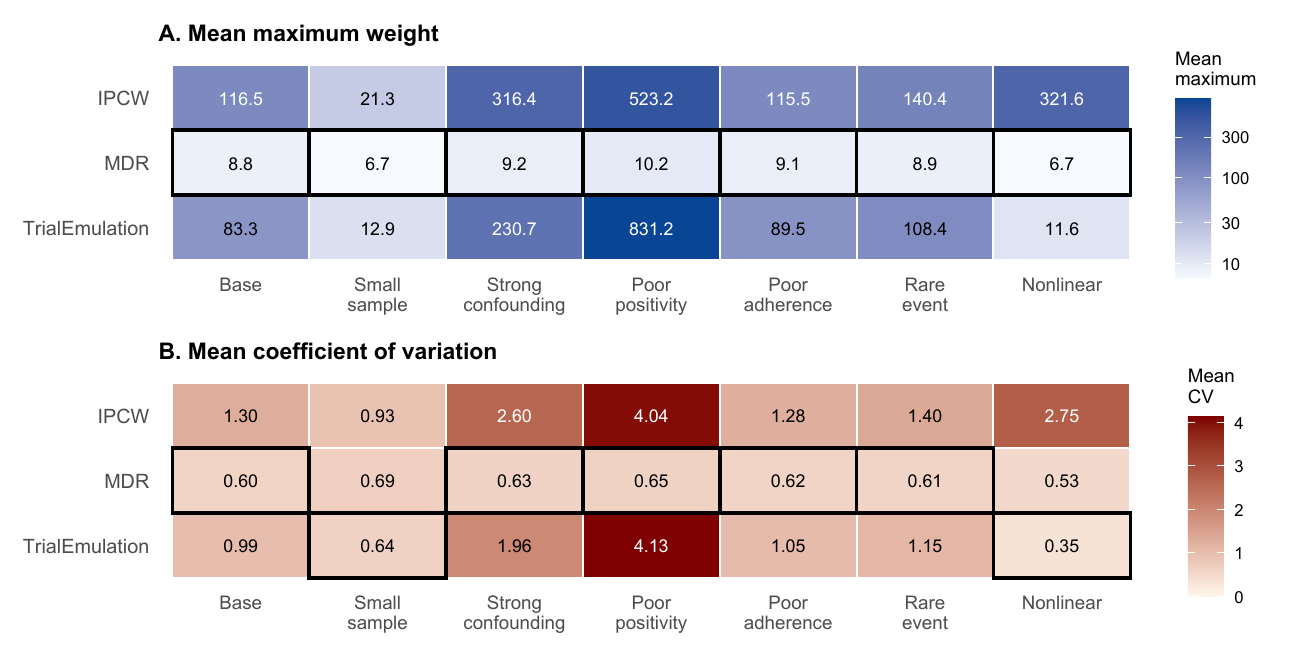}
  \caption{Stability of the analysis weights across the seven
  simulation scenarios. Entries are Monte Carlo means of the
  replication-specific maximum weight (Panel~A) and coefficient of
  variation (Panel~B), based on 500 replications per scenario. The
  color scale in Panel~A is logarithmic. Outlined cells identify the
  smallest value within each scenario. Each row represents one unique
  weight construction rather than every estimator using those weights.}
  \label{fig:appendix-weight-diagnostics}
\end{figure}

MDR had the smallest mean maximum weight in every scenario and its mean
coefficient of variation remained below \(0.70\). TrialEmulation had a
slightly smaller coefficient of variation than MDR in the small-sample
and nonlinear scenarios, but its mean maximum weight was larger. Both
IPCW and TrialEmulation showed pronounced weight instability under poor
positivity, whereas the MDR diagnostics changed little across
scenarios.

Figure~\ref{fig:appendix-weight-max-distribution} complements the Monte
Carlo means by showing the full distribution of the
replication-specific maximum weight. The logarithmic scale is needed
because the IPCW and TrialEmulation distributions are strongly
right-skewed.

\clearpage
\begin{figure}[!t]
  \centering
  \includegraphics[width=\textwidth]{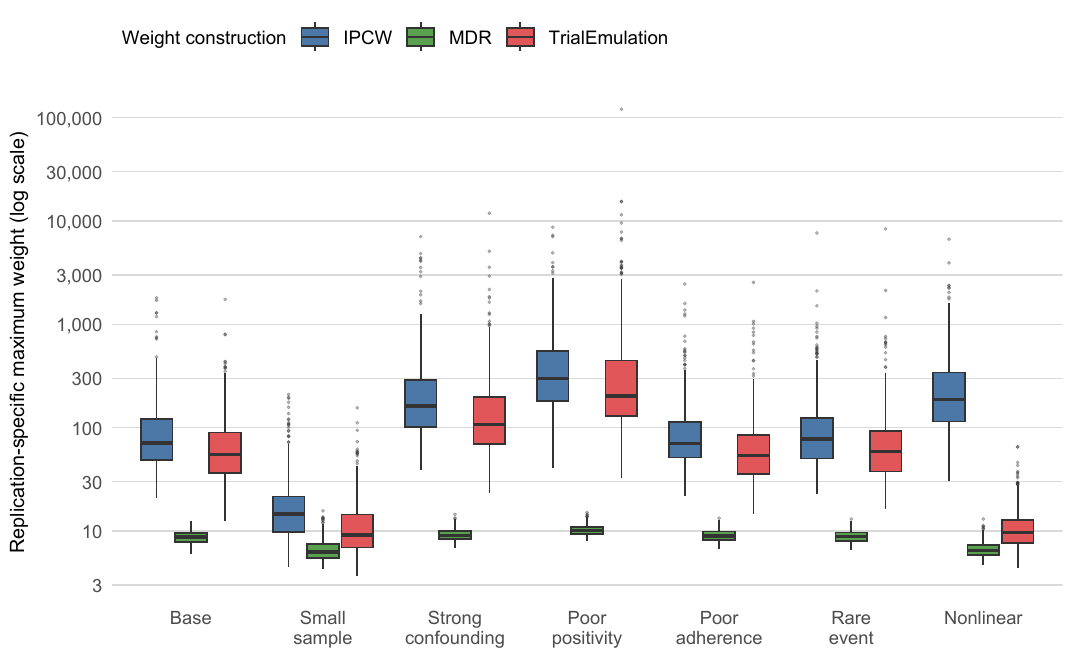}
  \caption{Distribution of the replication-specific maximum analysis
  weight across the seven simulation scenarios. Each boxplot summarizes
  500 replications; the box spans the interquartile range, the center
  line marks the median, the whiskers extend to the most extreme values
  within 1.5 interquartile ranges of the box, and points denote more
  extreme replications. The vertical axis is logarithmic. Colors denote
  the three unique weight constructions rather than every estimator
  using those weights.}
  \label{fig:appendix-weight-max-distribution}
\end{figure}

The MDR maximum weights were tightly concentrated in every scenario,
with medians from \(6.3\) to \(10.1\) and no replication-specific
maximum above \(15.7\). In contrast, IPCW and TrialEmulation had wide
upper tails in most scenarios. Under poor positivity, their median
maximum weights were \(299.0\) and \(201.7\), respectively, and the
most extreme maxima were \(8{,}710.4\) and \(120{,}863.6\), compared
with an MDR median of \(10.1\) and maximum of \(15.1\). Thus, the large
Monte Carlo mean maxima for IPCW and TrialEmulation reflect both higher
typical maxima and occasional extreme replications.

\clearpage
\begin{figure}[!t]
  \centering
  \includegraphics[width=\textwidth]{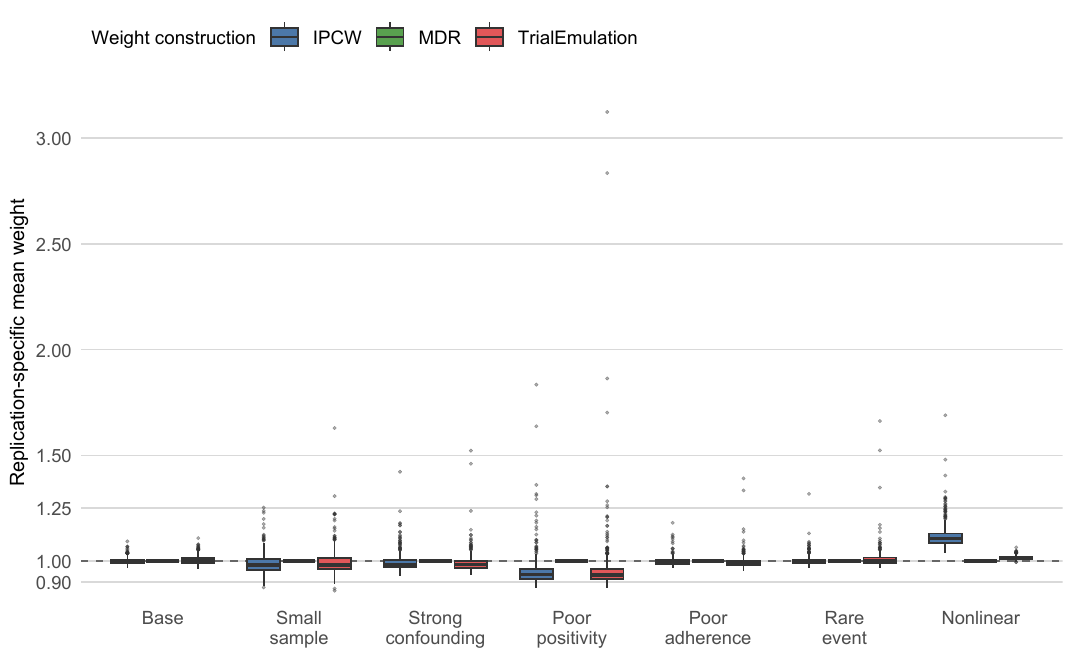}
  \caption{Distribution of the replication-specific mean analysis
  weight across the seven simulation scenarios. Each boxplot summarizes
  500 replications and follows the conventions described in
  Figure~\ref{fig:appendix-weight-max-distribution}. The vertical axis
  is linear, and the dashed horizontal line marks a mean weight of one.
  Colors denote the three unique weight constructions rather than every
  estimator using those weights.}
  \label{fig:appendix-weight-mean-distribution}
\end{figure}

By construction, the MDR mean weight equaled \(1.000\) in every
replication because normalization was performed separately within cells
defined by strategy and follow-up. IPCW and TrialEmulation also
generally had replication-specific means near one, although their
medians fell to \(0.935\) and \(0.934\), respectively, under poor
positivity. In the nonlinear scenario, the IPCW median rose to
\(1.106\), compared with \(1.015\) for TrialEmulation and \(1.000\) for
MDR. Occasional extreme replication means were also present: under poor
positivity, the largest mean was \(1.834\) for IPCW and \(3.124\) for
TrialEmulation. Thus, the mean-weight distributions reinforce the exact
normalization of MDR and reveal a small number of replications in which
the average IPCW or TrialEmulation weight departed substantially from
one.

%% file: scenario_diagnostics/table_scenario_diagnostics.tex
\begin{table}[!htbp]
\centering
\caption{Empirical characteristics of the simulation scenarios across 500 replications}
\label{tab:simulation-scenario-diagnostics}
\resizebox{\textwidth}{!}{%
\begin{tabular}{lrrrrrr}
\hline
Scenario &
\multicolumn{1}{c}{Participants} &
\multicolumn{1}{c}{$L_k\!\rightarrow\!A_k$ OR} &
\multicolumn{1}{c}{$L_k\!\rightarrow\!Y_{k+1}$ OR} &
\multicolumn{1}{c}{Extreme $p_A$ (\%)} &
\multicolumn{1}{c}{Continuation (\%)} &
\multicolumn{1}{c}{12-month risk (\%)} \\
\hline
Base & 2000 [2000, 2000] & 4.86 [4.73, 5.00] & 4.86 [4.73, 5.00] & 10.9 [10.3, 11.7] & 76.3 [75.5, 77.1] & 12.5 [11.0, 13.8] \\
Small sample & 100 [100, 100] & 4.86 [4.22, 5.52] & 4.86 [4.22, 5.52] & 10.8 [7.8, 14.4] & 76.1 [72.4, 79.5] & 12.1 [6.1, 19.6] \\
Strong confounding & 2000 [2000, 2000] & 9.92 [9.50, 10.38] & 9.92 [9.50, 10.38] & 21.9 [21.0, 22.8] & 75.3 [74.4, 76.1] & 16.1 [14.4, 17.7] \\
Poor positivity & 2000 [2000, 2000] & 43.51 [40.65, 46.84] & 4.52 [4.40, 4.66] & 40.3 [39.4, 41.4] & 75.1 [74.2, 76.0] & 12.2 [10.8, 13.6] \\
Poor adherence & 2000 [2000, 2000] & 4.89 [4.73, 5.04] & 4.89 [4.73, 5.04] & 5.5 [5.0, 6.1] & 55.5 [54.3, 56.7] & 12.6 [11.2, 14.1] \\
Rare event & 2000 [2000, 2000] & 4.88 [4.76, 5.02] & 4.88 [4.76, 5.02] & 11.2 [10.5, 11.9] & 76.5 [75.8, 77.3] & 1.1 [0.7, 1.6] \\
Nonlinear & 2000 [2000, 2000] & 3.86 [3.62, 4.15] & 3.13 [3.02, 3.27] & 13.7 [12.9, 14.6] & 81.8 [81.2, 82.4] & 15.7 [14.0, 17.2] \\
\hline
\end{tabular}%
}

\vspace{0.5em}
\parbox{\textwidth}{\footnotesize
Values are medians [empirical 2.5th, 97.5th percentiles] across 500
replications, except participants, which is median [minimum, maximum].
DGP denotes the data-generating process, and OR denotes odds ratio. The two
DGP-implied odds ratios compare the empirical third versus first quartile
of the time-varying confounder using the unclipped treatment and event DGP
linear predictors. Extreme $p_A$ is the percentage of monthly observations
with $p(A_k=1\mid\text{history})\leq0.05$ or $\geq0.95$. Continuation is
$\Pr(A_k=1\mid A_{k-1}=1)$. Cumulative event risk is the discrete-time
product-limit estimate through month 12.}
\end{table}